\documentclass[lettersize,journal]{IEEEtran}
\usepackage{amsmath,amsfonts, amssymb,bbm,amsthm,mathtools,nccmath}

\DeclareMathOperator*{\argmin}{\arg\!\min}
\usepackage{algorithmic}
\usepackage{algorithm}
\usepackage{array}
\usepackage[caption=false,font=normalsize,labelfont=sf,textfont=sf]{subfig}
\usepackage{textcomp}
\usepackage{stfloats}
\usepackage{url}
\usepackage{verbatim}
\usepackage{graphicx}
\usepackage{enumitem}
\usepackage{cite}
\usepackage[normalem]{ulem}
\usepackage{cancel}
\usepackage{xcolor}
\usepackage{hyperref}

\newtheorem{theorem}{Theorem}\newtheorem{lemma}{Lemma}\newtheorem{remark}{Remark}\newtheorem{corollary}{Corollary}
\newtheorem{proposition}{Proposition}\newtheorem{definition}{Definition}\newtheorem{assumption}{Assumption}

\begin{document}

\title{Age of Incorrect Information With Hybrid ARQ Under a Resource Constraint for $N$-ary Symmetric Markov Sources}

\author{Konstantinos Bountrogiannis,~Anthony Ephremides,~Panagiotis Tsakalides,~and~George Tzagkarakis% <-this % stops a space
\IEEEcompsocitemizethanks{\IEEEcompsocthanksitem Konstantinos Bountrogiannis and Panagiotis Tsakalides are with the Department of Computer Science, University of Crete, Heraklion 700~13, Greece, and with the Institute of Computer Science, Foundation for Research and Technology -- Hellas, Heraklion 700~13, Greece\protect\\
E-mail: kbountrogiannis@csd.uoc.gr;  tsakalid@csd.uoc.gr
\IEEEcompsocthanksitem Anthony Ephremides is with the Electrical and Computer Engineering Department, University of Maryland, College Park, MD 20742 USA\protect\\
E-mail: etony@umd.edu
\IEEEcompsocthanksitem George Tzagkarakis is with the Institute of Computer Science, Foundation for Research and Technology -- Hellas, Heraklion 700~13, Greece\protect\\
E-mail: gtzag@ics.forth.gr}% <-this % stops a space
\thanks{This work was supported in part by the European Commission under Grant 101094354 (ARGOS - Conceptual Design Study project) and in part by the European Commission under the framework of the National Recovery and Resilience Plan Greece 2.0 -- NextGenerationEU, Grant TAEDR-0536642 (Smart Cities project).}}

%\markboth{IEEE/ACM Transactions on Networking}%
%{Bountrogiannis \MakeLowercase{\textit{et al.}}: Age of Incorrect Information With Hybrid ARQ Under a Resource Constraint for $N$-ary Symmetric Markov Sources}

%\IEEEoverridecommandlockouts
%\IEEEpubid{\begin{minipage}{\textwidth}\copyright\ 2024 IEEE. Personal use of this material is permitted. Permission from IEEE must be obtained for all other uses, in any current or future media, including reprinting/republishing this material for advertising or promotional purposes, creating new collective works, for resale or redistribution to servers or lists, or reuse of any copyrighted component of this work in other works. \hfill \hspace{\columnsep}\end{minipage}}

\maketitle

%\IEEEpubidadjcol

\thanks{This work is accepted for publication in the IEEE/ACM Transactions on Networking. This is the final preprint version. For the published version, refer to \url{https://ieeexplore.ieee.org/abstract/document/10767719}}

\begin{abstract}
The Age of Incorrect Information (AoII) is a recently proposed metric for real-time remote monitoring systems. In particular, AoII measures the time the information at the monitor is incorrect, weighted by the magnitude of this incorrectness, thereby combining the notions of freshness and distortion. This paper addresses the definition of an AoII-optimal transmission policy in a discrete-time communication scheme with a resource constraint and a hybrid automatic repeat request (HARQ) protocol. Considering an $N$-ary symmetric Markov source, the problem is formulated as an infinite-horizon average-cost constrained Markov decision process (CMDP). Interestingly, it is proved that, under some conditions, the optimal transmission policy is to never transmit. This reveals a region of the source dynamics where communication is inadequate in reducing the AoII. Elsewhere, there exists an optimal transmission policy, which is a randomized mixture of two discrete threshold-based policies that randomize on at most one state. The optimal threshold and the randomization component are derived analytically. Numerical results illustrate the impact of the source dynamics, channel conditions, and resource constraints on the average AoII.
\end{abstract}


%\ifCLASSOPTIONcompsoc
\section{Introduction}\label{sec:introduction}

\IEEEPARstart{T}{he} technological advancements in sensor and monitoring devices, together with the development and widespread utilization of the 5G cellular networks and beyond, lead to the continuous emergence of new applications, whose principal element is the real-time monitoring of remote sources. The increasing list of examples includes autonomous driving, real-time video feedback, anomaly detection in critical infrastructures, remote surgery, emerging applications in augmented reality networks and haptic communications. In such applications, timely delivery of information is fundamental.

It is well understood that, while low-latency networks are necessary, they are insufficient to guarantee timely operation~\cite{bib:aoi}. This has increased the interest in the Age of Information (AoI) metrics to analyze and design such real-time applications. This new family of communication metrics captures the end-to-end latency in remote monitoring systems. Principally, the instantaneous AoI at time $t$ is defined as the difference $t\!-\!u_t\!\geq\! 0$, where $u_t$ is the generation time of the most recently successfully decoded packet. Therefore, AoI quantifies the freshness of the information content of a packet and the importance of updating the monitor with fresh information due to excess ageing.

The most important contribution of AoI has been the opening of a new perspective in the analysis and design of task-oriented communication systems. Since its introduction in~\cite{bib:aoi}, AoI has attracted the interest of researchers and engineers from many fields~\cite{bib:aoi_survey}. Nevertheless, a shortcoming of the conventional AoI metric was shortly noticed. In particular, AoI quantifies the information freshness but omits the dynamics of the data source. For example, consider a source that changes rapidly and another that changes slowly. The packets that are generated simultaneously from the two sources will have the same AoI, but obviously, the packet from the rapidly-changing source is less accurate. The initial reaction of the research community pointed towards the generalization of the conventional AoI metric with non-linear age functions~\cite{bib:nonlinear_Kosta, bib:nonlinear_Sun, bib:value}, or even arbitrary non-decreasing functions~\cite{bib:updatewait}. In~\cite{bib:stamatakis}, the authors extended AoI to measure the time elapsed since the generation of novel source values.

From another perspective, a rather expected yet fascinating observation has been made, i.e., minimizing the AoI is not necessarily equivalent to minimizing the real-time estimation error. Particularly, if the sampling times are independent of the observed source, it can be shown that the mean-squared estimation error is an increasing function of the AoI. However, if the sampling times depend on the history of the source, the estimation error is not necessarily minimized with the AoI. This was shown to happen even in the simplest signals~\cite{bib:wiener, bib:ornstein}, and similar results were derived for the real-time state estimation error of feedback control systems~\cite{bib:state1, bib:state2}.

The conventional AoI ignores the content of the communicated data. On the other hand, the traditional error metrics do not capture the amount of time the monitor's estimate is erroneous. To address these limitations, the Age of Incorrect Information (AoII) was proposed in~\cite{bib:aoii}. Particularly, the AoII is defined as an age function weighted by the real-time information mismatch (distortion), where the age function penalizes only the time that the information mismatch is non-zero. Essentially, AoII measures the time that the information at the monitor is incorrect, weighted by the magnitude of this incorrectness, rendering AoII a semantic metric that measures both the timeliness and the accuracy of the delivered information. Moreover, AoII is a suitable metric for task-oriented communications, where the age and the distortion functions can be naturally specified by the application of interest. A visual comparison of the AoII and AoI metrics is shown in Fig.~\ref{fig:aoi_vs_aoii}.

\begin{figure}
    \centering
    \includegraphics[width=.94\linewidth]{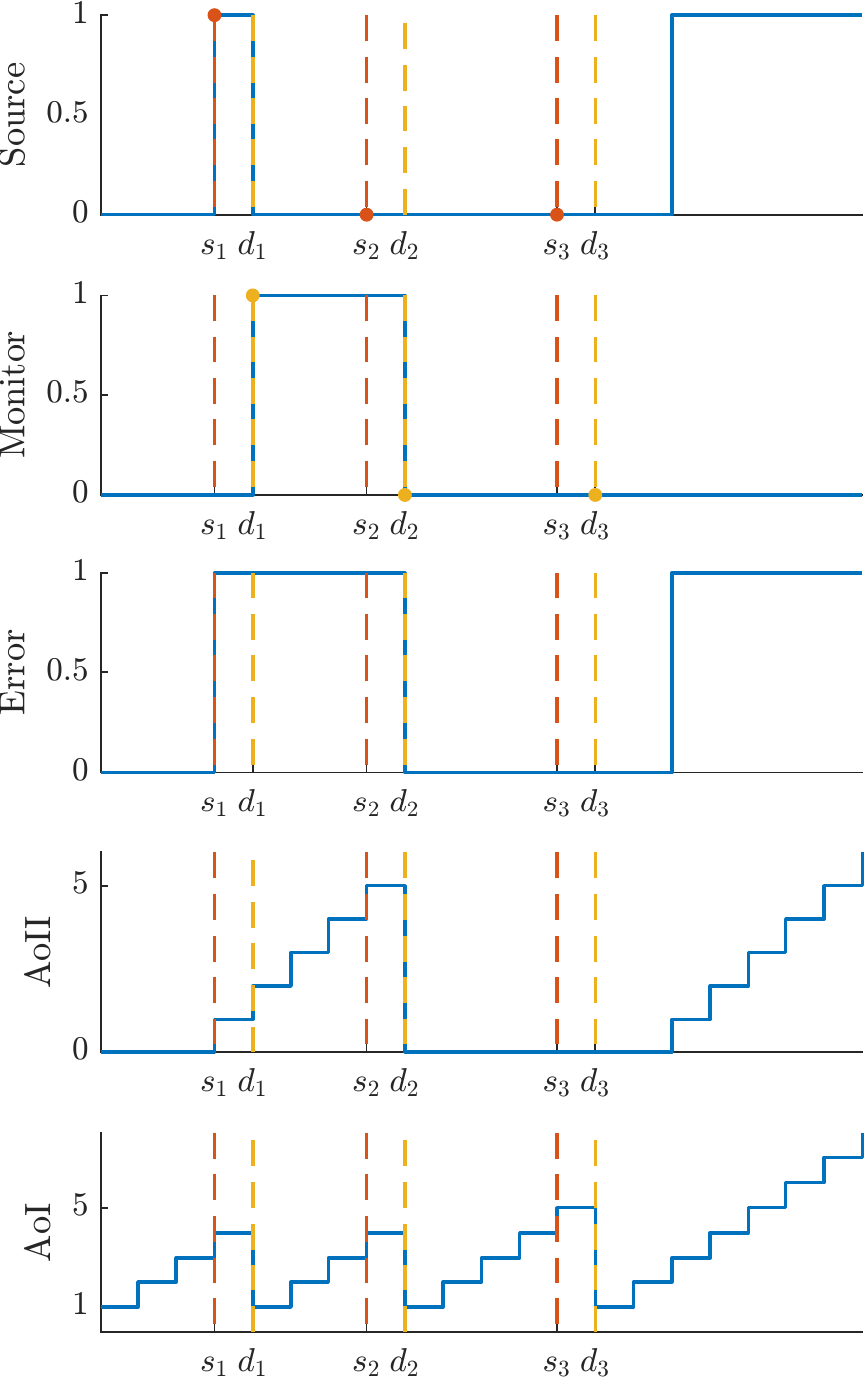}
    \caption{Snapshot of a communication system illustrating the differences between AoI and AoII. In this example, a binary \textit{source} is sampled at every time slot and is transmitted to the \textit{monitor}. The packet arrives at the next time slot and is attempted to be decoded. Successful decodings occur at time slots $d_i$, $i=1,2,3$, whereas the respective samples were generated at slots $s_i=d_i-1$. If the decoding is successful, the \textit{monitor} updates its estimate and the \textit{AoI} decreases to one (because the sample was generated one time slot ago). Otherwise, the \textit{AoI} increases. The corresponding \textit{error} indicates the mismatch between the \textit{source} and the \textit{monitor} estimate. The \textit{AoII} measures the time slots where the \textit{error} has been positive.}
    \label{fig:aoi_vs_aoii}
\end{figure}

Several works have been published since the introduction of AoII. In~\cite{bib:aoii}, the authors consider a simple indicator distortion function and study the minimization of AoII under resource constraints, whilst in~\cite{bib:aoii_discrete}, the distortion function has multiple thresholds. In~~\cite{bib:aoii_semantics2} the results are generalized with task-oriented age functions. In~\cite{bib:aoii_pl}, the authors analyze the average AoII for a piecewise linear signal where the transmitter updates the monitor for slope changes. The authors in~\cite{bib:aoii_multiple_sources2} consider a system with multiple sources where the scheduler is at the side of the receiver. In~\cite{bib:aoii_multiple_sources1}, a similar problem is studied where the scheduler has imperfect channel state information. The work in~\cite{9686027} considers a resource-constrained scheme with a Markov source where transitions happen only linearly, i.e. between adjacent states.

Our work focuses on a discrete-time (slotted) communication system, where each packet has a probability of being successfully decoded. The transmitter and the receiver employ a hybrid automatic repeat request (HARQ) protocol to correct communication errors. In particular, the packets are encoded using a forward error-correction code to correct communication errors at the receiver. If the decoder fails to decode the packet, it requests a re-transmission with a NACK feedback message. When HARQ is used with soft combining, the decoder combines all the received packets to improve the probability of successful decoding~\cite{bib:harq_survey,bib:harq_fading}. In the case of soft combining, a re-transmission can either consist of an identical packet (chase combining HARQ) or some complementary information to the previously transmitted packets (incremental redundancy HARQ)~\cite{bib:harq_survey}. The standard ARQ protocol is different in that the decoder only detects errors but cannot correct them, and hence discards the previous packets and considers only the most recent ones. This is similar to HARQ without soft combining but without the ability to correct any error. An important aspect of HARQ with soft combining is that the probability of successful decoding increases with the total number of packets, while with the simple HARQ and standard ARQ, it is constant.

In this work, we develop optimal scheduling policies for minimizing the average AoII in a communication system with HARQ under a resource constraint. The resource constraint is motivated by limitations on power or network resourc,
es (e.g. battery-powered sensors, allocated resources in sensor networks, etc.). The same constraint was considered in~\cite{bib:aoii_discrete, bib:aoii, bib:aoii_semantics2} for the minimization of the AoII without HARQ. A related problem is considered in~\cite{bib:aoi_harq}, where the authors minimize the simple AoI with HARQ under the same constraint. However, the authors only give an analytical solution for the standard ARQ and approximate the solution for general HARQ protocols. Notably, the analysis of AoII is a harder task than that of the traditional AoI. This stems from the inclusion of the communicated information content in the metric via the distortion function.

It is noteworthy that age-related optimal transmission policies in resource-constrained environments commonly exhibit a threshold structure, where transmissions occur only if the age exceeds a certain threshold~\cite{bib:aoi_harq,bib:aoii,bib:threshold_policy_3,bib:threshold_policy_5,bib:threshold_and_uniform_policy}. This property is iterated in our current work. Nevertheless, the threshold structure is not unique in such environments. For instance,~\cite{bib:threshold_and_uniform_policy, bib:uniform_policy_1,bib:uniform_policy_3} demonstrate the optimality of best-effort policies, where transmissions occur at uniformly spaced intervals.
%It is worth noting that age-related optimal transmission policies in resource-constrained environments commonly exhibit a threshold structure, a property that is iterated in our current work. The online policy in~\cite{bib:threshold_policy_1} is examined within a system featuring an energy harvesting sensor equipped with an infinite-sized battery, demonstrating that the policy follows a threshold structure. The same is established in~\cite{bib:threshold_policy_3} for finite-sized batteries. Additionally, the optimality of these policies under various age penalty functions is confirmed in~\cite{bib:threshold_policy_5}. For systems with unit-sized batteries and update erasures, threshold structure is also evident in online policies, as shown in~\cite{bib:threshold_policy_6}. Nevertheless, the threshold structure is not unique in such environments. For instance, the study in~\cite{bib:threshold_and_uniform_policy} demonstrates that best-effort online policies, where updates are sent at uniformly spaced intervals if energy permits, are optimal for systems with infinite batteries. Furthermore,~\cite{bib:uniform_policy_1} and~\cite{bib:uniform_policy_2} extend this finding to multihop networks for online and offline policies, respectively. Best-effort is also shown to be optimal for the infinite battery case when updates are subject to erasures, in~\cite{bib:uniform_policy_3}. The above illustrates that the optimal policy structure depends on other parameters of the problem, some of which seem to be the sampling rate

To the best of our knowledge, this is the first study that addresses the analysis of the AoII in a communication system with HARQ. The main contributions of this paper are summarized as follows: 
\begin{itemize}[leftmargin=*]
    \item We analyze AoII in a communication scheme utilizing HARQ within the confines of a resource constraint that limits the long-term average transmission rate. We explore HARQ with and without soft combining. The source model employed in this context is an $N$-ary symmetric Markov source.
    \item The transmission policy optimization problem is framed as a challenging task within the realm of infinite-horizon average-cost constrained Markov decision processes (CMDPs), which are typically known for their difficulty in achieving exact solutions. Nonetheless, an exact solution is attained by examining the inherent structural properties of the CMDP.
    \item When considering an HARQ protocol, the state space of the CMDP becomes a union of two countably infinite sets due to tracking the AoII and the transmission count. Handling infinite state spaces poses additional analytical challenges, as documented in~\cite[Sec. 4.6]{bib:bertsekas}. The primary technical complexity, different from prior research, arises from the interdependence of these variables, making it challenging to establish the cost monotonicity. Consequently, the proof of cost monotonicity employed in threshold-based solutions % (such as \cite{9518209,8437712, bib:aoi_harq, bib:aoii})
    necessitates a deeper examination of the underlying state transitions within the context of our current work.
    Additionally, for multivariate-state CMDPs, the threshold is an arbitrary function of the other variables, making it generally difficult to find an exact solution. However, we found an analytical solution, which sets our approach apart from other studies that focus on leveraging the MDP structure to enhance the convergence speed of approximate algorithms, as exemplified in references~\cite{bib:aoi_harq,8778671,8972306,9085402}.
    %Compared to the work in~\cite{bib:aoi_harq}, where the problem is solved for the AoI metric, we provide an analytical solution for any HARQ protocol, not only for the standard ARQ. This is also in contrast with other related works that employ the structure of the MDP to increase the convergence rate of approximate algorithms, rather than finding an exact solution, e.g.~\cite{8778671,8972306,9085402}.
    \item Interestingly, we demonstrate that, given certain conditions, the optimal policy is to abstain from transmission entirely. In all other scenarios, we establish that the optimal approach involves a randomized mixture of two distinct threshold-based policies. It is shown that the thresholds are independent of the packet count, further simplifying the optimal policy. We analytically deduce the precise threshold value and the randomization component that align with the resource constraints.
    \item Extensive simulations are conducted to investigate how the average AoII is influenced by the source dynamics, channel conditions, and resource constraints when employing the optimal policy.
\end{itemize}

The rest of this paper is organized as follows: in Section~\ref{sec:Problem}, the system model is presented and the problem is formulated as a CMDP. Then, the constrained problem is expressed as an unconstrained Lagrangian MDP. Section~\ref{sec:Lagrangian} analyzes the structural properties of the Lagrangian MDP and derives the Lagrange-optimal transmission policy. Section~\ref{sec:CMDP} shows that the optimal policy of the constrained problem is a randomized mixture of two Lagrange-optimal policies. Section~\ref{sec:Lagrangian_proofs} provides the proofs of the structural properties given in Sec.~\ref{sec:Lagrangian}. Section~\ref{sec:algorithm} describes an efficient algorithm that computes the optimal policy, and Section~\ref{sec:results} gives numerical results of the average AoII under the optimal policy with varying model parameters. Finally, Section~\ref{sec:conclusion} summarizes the main outcomes and presents directions for further extensions.

\section{Problem Definition}\label{sec:Problem}

\subsection{Communication model}

We consider a discrete-time (slotted) communication model over a noisy channel. The transmitter monitors a data source from which fresh samples arrive at every time slot. At each time slot, the transmitter decides whether to transmit or discard the fresh sample. The samples are encoded with a channel coding scheme and sent through the channel to the receiver. Upon receiving the packet, the receiver attempts to decode it. If the decoding is successful, the receiver notifies the transmitter with an ACK packet. Otherwise, it sends a NACK message to ask for additional information. The transmitter decides whether it sends additional information or rejects the request.

We assume that the channel states of each time slot are independent and identically distributed. Furthermore, we assume that the duration of a packet transmission is constant and equal to one time slot, whilst the ACK/NACK packets are instantaneous. The instantaneous feedback message is a typical assumption in the literature, justified by its limited information content. The constant delay has also been adopted in many other works (e.g.~\cite{bib:aoi_harq, bib:aoii, bib:constant_delay_1, bib:constant_delay_2, bib:constant_delay_3, bib:constant_delay_4}). Note that in typical networks there are two major sources of randomness in the delay: a) queues formed by packet congestion in relay nodes and b) erroneous packets that need to be re-transmitted. In our work, we abstract the stochasticity of the delay due to erroneous packets and embed it in a model with a constant delay and a specified probability of successful decoding. Therefore, our model is close to reality when the randomness incurred by queues is negligible.

The probability of successful decoding at each time slot is specified by a non-decreasing function $p(r)\in (0,1)$, where $r\in \mathbbm{N}$ is the number of packets already gotten by the receiver. For example, the probability of successful decoding for the first packet is given by $p(0)$, for the second packet is given by $p(1)$, and so on. Typically, the maximum number of packets is limited and only a total of $r_{max}$ re-transmissions are allowed~\cite{bib:harq_survey}.

Our objective is to minimize the average AoII. In particular, let $g(X_t,\hat{X_t})$ denote the distortion function at time $t$ between the source $X_t$ and its estimation $\hat{X_t}$ at the receiver. In general, $g(X_t,\hat{X_t})$ is directed by the specific application of interest. Moreover, the age function is defined as
\begin{equation}
    \Delta_t \triangleq t-h_t\,,
\end{equation}
where $h_t$ is the last time instant when the distortion $g(X_t,\hat{X_t})$ was zero. The instantaneous AoII at time $t$ is simply the product of the age and distortion functions,
\begin{equation}\label{eq:aoii_def}
    \delta_t \triangleq \Delta_t \cdot g(X_t,\hat{X_t}).
\end{equation}
We shall generalize the AoII process in~\eqref{eq:aoii_def} by employing a generic monotonically increasing\footnote{The condition of strict monotonicity is a bit more restrictive than the weak monotonicity adopted in~\cite{bib:aoii_semantics2, bib:aoii_multiple_sources1}, but simplifies the proofs of our results.} and unbounded penalty function $f(\delta_t)$, i.e., $f(\delta')>f(\delta)\ \forall\ \delta'>\delta$ and $\lim_{\delta\to+\infty}f(\delta)=+\infty$.

Lastly, motivated by requirements on saving or allocating power and network resources, we impose a constraint on the transmission rate by requiring the long-term transmission rate to not exceed $R\in(0,1]$.

\subsection{Source Model}

Hereafter, we focus on $N$-ary symmetric Markov sources, as illustrated in Fig.~\ref{fig:source}. Here, $\alpha$ is the probability that the source remains at the same state at the next time slot, and $\mu$ is the probability of transition to all other $N-1$ states. The same model is considered in~\cite{bib:aoii} and~\cite{bib:aoii_multiple_sources2}.

\begin{figure}
    \centering
    \includegraphics[width=\linewidth]{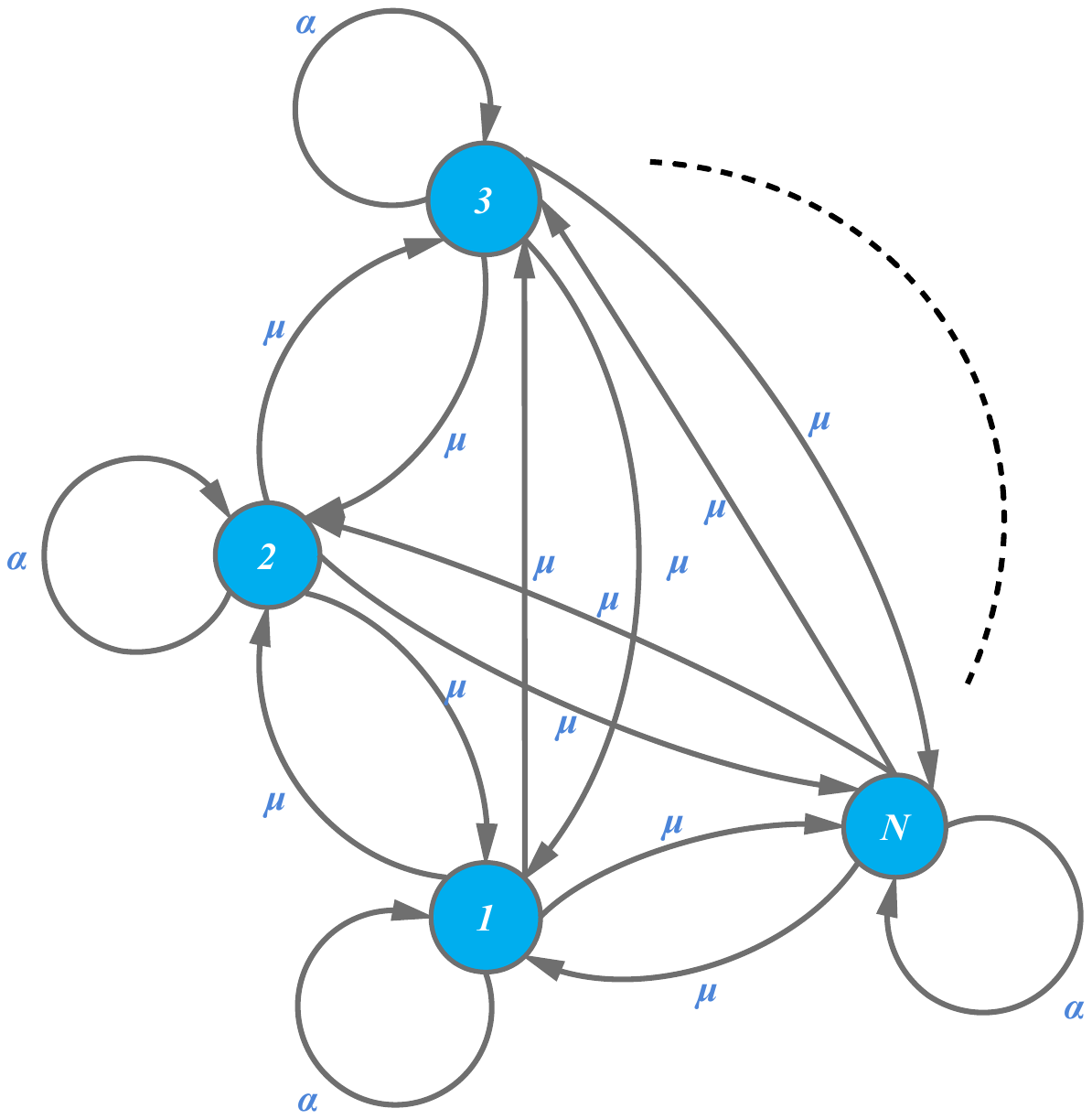}
    \caption{The symmetric Markov data source under consideration.}
    \label{fig:source}
\end{figure}

Due to the normalization property of the transition probabilities, the following equality must hold,

\begin{equation}\label{eq:mu_alpha_relation}
    (N-1)  \mu + \alpha = 1.
\end{equation}
In our analysis, we employ an indicator distortion function defined as follows,

\begin{equation}
    g(X_t,\hat{X_t})=\mathbbm{1}_{\{X_t\neq \hat{X_t}\}}\,.
\end{equation}
This distortion function penalizes equally any information mismatch between the source and the monitor. The symmetric Markov source implies that the source changes (or gives new information) every $Y$ time slots, where $Y$ is geometrically distributed, with the same parameter for all states. Note that the indicator cost function may also arise by truncating other more complex distortion functions, i.e., it can model the function $\mathbbm{1}_{\{d(X_t, \hat{X_t})>c\}}\,$, where $d(X_t, \hat{X_t})$ is any distortion function and $c$ is a specified threshold. In this case, we lose some information but gain the tractability of the optimization problem.

Although we focus on finite-state Markov chains, our results also hold for infinite-state Markov chains. This adds to the versatility of our results. We note two interesting cases in this matter. Firstly, suppose that $\alpha\!=\!0$ and $N\!\rightarrow\!\infty$. Thus, the source gives new information at every time slot, but the information at two different time slots is never the same. This implies that the AoII falls back to the simple AoI. Secondly, suppose that the source is a finite-state Markov chain with a fixed probability of staying at the same state $\alpha\!>\!0$. Now construct another Markov chain that is infinite-state but with the same probability of staying at the same state $\alpha$. We have $\alpha\!>\!0,\ N\!\rightarrow\!\infty,\ \mu\!\rightarrow\! 0$. With the constructed Markov chain, once the source departs from a state, the probability of returning to it in the future is zero. This reflects the scenario where the AoII penalty is zero if there is no new information but always increases otherwise, even if the real source eventually returns to the known value. This model can be used when the generation time of the information is part of the information itself.

\subsection{Mathematical Formulation}\label{sec:MDP}

We formulate the problem as an infinite-horizon average-cost constrained Markov decision process (CMDP). The next subsections define the CMDP and present the necessary assumptions.

\subsubsection{Definition of the Constrained Markov Decision Process}

Before we define the CMDP, we highlight a remark that simplifies the definition of the problem. In particular, the transmitter is not obliged to fulfil the request for a re-transmission. However, should a re-transmission occur, it is better to do it immediately after receiving the NACK message than wait some time. This is true since waiting i) increases the age and ii) incurs the chance of the packet becoming obsolete due to a change of the source, all without increasing the probability of decoding. A similar observation has been made in the study of AoI with HARQ~\cite{bib:aoi_harq}.

\begin{remark}\label{rem:wait_count}
    The AoII-optimal transmission policy incurs an HARQ re-transmission only immediately after the reception of a NACK feedback message.
\end{remark}

Therefore, we presume that if the transmitter decides to not fulfil the re-transmission request immediately, the request is rejected altogether and the transmission count $r$ reverts to zero.

That being so, the CMDP is defined as follows:
\begin{itemize}[leftmargin=*]
    \item The state of the CMDP at time $t$ is given by $S_t=(\delta_t,r_t) \in \mathcal{S}$, where $\delta_t$ is the AoII at the current time slot and $r_t$ is the transmission count for the current source state.
    \item The cost at state $S_t$ is equal to the instantaneous AoII penalty $f(\delta_t)$.
    \item The actions $y_t\in \mathcal{Y}$, where the action space $\mathcal{Y}=\{0,1\}$ consists of the ``wait'' ($y_t=0$) and the ``transmit'' ($y_t=1$) actions.
    \item Define the functions $\gamma_1(r_t)$, $\gamma_2(r_t)$ as
    \begin{align}
        &\gamma_1(r_t) \triangleq \alpha (1-p(r_t)),\label{eq:gamma1}\\
        &\gamma_2(r_t) \triangleq 1-\alpha - \mu(1-p(r_t)).\label{eq:gamma2}
    \end{align}
    
    The transition probabilities are summarized as follows, while a detailed derivation is given in the Appendix.
    ~\\
    
    If action is ``wait'' ($y_t=0$):
    \begin{equation}\label{eq:transition_y0}
        \begin{aligned}
        &P(S_{t+1}{=}(\delta_t{+}1,0)\mid\delta_t>0,r_t\geq 0, y_t = 0) = 1-\mu\,,\\
        &P(S_{t+1}{=}(0,0)\mid\delta_t>0,r_t\geq 0, y_t = 0) = \mu\,,\\
        &P(S_{t+1}{=}(1,0)\mid\delta_t=0,r_t= 0, y_t = 0) = 1-\alpha\,,\\
        &P(S_{t+1}{=}(0,0)\mid\delta_t=0,r_t= 0, y_t = 0) = \alpha\,.\\\omit\hfill
        \end{aligned}
    \end{equation}
    If action is ``transmit'' ($y_t=1$):
    \begin{equation}\label{eq:transition_y1}
        \begin{aligned}
        &P(S_{t+1}{=}(\delta_t{+}1,r_t{+}1)\mid\delta_t>0,r_t\geq 0, y_t = 1) = \gamma_1(r_t)\,,\\ 
        &P(S_{t+1}{=}(\delta_t{+}1,0)\mid\delta_t>0,r_t\geq 0, y_t = 1) = \gamma_2(r_t)\,,\\
        &P(S_{t+1}{=}(0,0)\mid\delta_t\!>\!0,r_t\!\geq\! 0, y_t\!=\!1) = 1\!-\!\gamma_1(r_t)\!-\!\gamma_2(r_t)\,,\\
        &P(S_{t+1}{=}(1,0)\mid\delta_t=0,r_t= 0, y_t = 1) = 1-\alpha\,,\\
        &P(S_{t+1}{=}(0,0)\mid\delta_t=0,r_t= 0, y_t = 1) = \alpha\,.
        \end{aligned}
    \end{equation}
    \item The long-term average number of ``transmit'' actions is constrained to not exceed $R\in(0,1]$.

\end{itemize}

\subsubsection{Additional Assumptions}

To derive our results, it is necessary to impose the following condition, which ensures that the average AoII is finite under the policy that a new transmission is performed in every time slot,

\begin{equation}\label{eq:bound_cond}
    \sum_{l=1}^\infty f(l{+}1)(\gamma_1(0)+\gamma_2(0))^{l} < \infty\,.
\end{equation}

Furthermore, without loss of generality, we assume an unlimited maximum number of allowed re-transmissions, i.e., $r_{max}=\infty$. In Section~\ref{sec:results}, we impose a finite $r_{max}$ without affecting the theoretical results.

\subsubsection{Definition of the Optimization Problem}

The optimization problem pertains to finding the policy $\pi:\mathcal{S}\mapsto\mathcal{Y}$ that minimizes the long-term average AoII, while not exceeding the transmission rate constraint. The problem can be expressed as a linear programming problem, as follows,

\begin{definition}[Main CMDP Problem]
\begin{equation}\label{eq:opt_problem}
\begin{aligned}
    \text{Minimize}\quad &J_\pi(S_0) \!\triangleq\! \limsup\limits_{T\rightarrow \infty}\frac{1}{T}E_{\pi}\!\left[\sum_{t=0}^{T-1}f(\delta_t)\mid S_0\right]\\
    \text{subject to}\quad &C_\pi(S_0) \!\triangleq\! \limsup\limits_{T\rightarrow \infty}\frac{1}{T}E_{\pi}\!\left[\sum_{t=0}^{T-1}y_t\mid S_0\right]\!\leq\! R\,.
\end{aligned}
\end{equation}
\end{definition}

To solve the constrained problem, we introduce the Lagrangian average cost and solve the relaxed problem,

\begin{definition}[Lagrangian MDP Problem]
\begin{equation}\label{eq:lagrange_problem}
    \text{Minimize}\; J(\pi{,}\lambda) {\triangleq} {\lim\limits_{T\rightarrow \infty}}\!\sup\limits_{\lambda\geq 0}\!\frac{1}{T}E_{\pi}{\left[\sum_{t=0}^{T-1}f(\delta_t){+}\lambda y_t\mid S_0\right]} {-} \lambda R\,.
\end{equation}
\end{definition}
\noindent For any fixed value of $\lambda$, let
\begin{equation}
    \pi_\lambda \triangleq \argmin_\pi J(\pi,\lambda)\,,
\end{equation}
\begin{equation}
    g \triangleq \min_\pi J(\pi,\lambda)\,,
\end{equation}
denote the optimal policy of the Lagrangian MDP and the average cost achieved by the optimal policy, respectively.

The following lemma provides the means for an alternative mathematical formulation.
\begin{lemma}\label{lem:unichain}
The MDP~\eqref{eq:lagrange_problem} is unichain. That is, there exists a single recurrent class and a (possibly empty) transient class.
\end{lemma}

\begin{proof}
    From the transition probabilities~\eqref{eq:transition_y0}, \eqref{eq:transition_y1}, we observe that, under any policy, there is a positive probability from every state to transit to state $S=(0,0)$. This suffices to prove the statement.% In addition, starting from the state $S=(0,0)$, there is a probabilistic path to every other state. These two conditions prove the statement.
\end{proof}

Due to Lemma~\ref{lem:unichain} and~\cite[Thm. 6.5.2]{bib:krishnamurthy}, the optimal policy $\pi_\lambda$ can be found by solving the following Bellman equations,

\begin{equation}\label{eq:bellman_def}
    g + V(S_t) = \min_{y_t}\{f(\delta_t)+\lambda y_t+\sum_{S_{t+1}}P(S_{t+1}\!\mid\! S_t, y_t)V(S_{t+1})\}.
\end{equation}
The function $V(S_t)=V(\delta_t,r_t)$ is called the \emph{value function} of state $S_t$.

In the following, Section~\ref{sec:Lagrangian} elaborates on the Lagrangian MDP, and Section~\ref{sec:CMDP} reverts to the main CMDP to define the optimal policy. The proofs of the results of Sec.~\ref{sec:Lagrangian} are written in Sec.~\ref{sec:Lagrangian_proofs}, while the results of Sec.~\ref{sec:CMDP} are given in the Appendix.

\section{Structural Results For The Lagrangian Problem}~\label{sec:Lagrangian}

This section gives the structural properties of the Lagrangian MDP. All results given here are proved in Section~\ref{sec:Lagrangian_proofs}.

The following two lemmas characterize how varying $\delta$ or $r$ affects the value function $V(\delta,r)$ and are necessary for proving the structural property of the optimal transmission policy.
\begin{lemma}\label{lem:value_delta}
The function $V(\delta,r)$ is increasing w.r.t. $\delta$.
\end{lemma}

\begin{lemma}\label{lem:value_r}
If $\mu<\alpha$, the function $V(\delta,r)$ is non-increasing w.r.t. $r$. If $\mu\geq\alpha$, it is non-decreasing.
\end{lemma}

Our first main result is the following proposition:
\begin{proposition}\label{prop:threshold}
Given a fixed transmission count $r$, if $\mu<\alpha$, the optimal policy at state $S=(\delta,r)$ is threshold-based w.r.t. $\delta$. If $\mu\geq\alpha$, the optimal policy is to always wait.
\end{proposition}

Proposition~\ref{prop:threshold} reveals that when the probability $\alpha$ is too small, communication is inadequate in reducing the AoII. This is a natural consequence since a transmission is useful only when the most likely state at the time of reception is the one being transmitted. Notice that a more sophisticated estimator at the receiver would leverage the received state to estimate a different state as the most probable. However, this would require the knowledge of the source dynamics at the side of the receiver, which is a rather strict assumption.

As a corollary of Proposition~\ref{prop:threshold}, the average AoII for $\mu\geq\alpha$ can be found by solving for $g$ in~\eqref{eq:bellman_def}, which is easy to do since the optimal policy is $y=0$ for all steps.

\begin{corollary}
If $\mu\geq\alpha$, the average AoII under the optimal (waiting) policy is given by
\begin{equation}\label{eq:waiting_AoII}
    g_{wait} = \frac{\mu\left(f(0)+(1-\alpha)\sum_{i=1}^\infty (1-\mu)^i f(i)\right)}{(1-\alpha)(1-\mu)+\mu}\,.
\end{equation}
\end{corollary}
%\begin{proof}
%    The detailed derivation can be found in Appendix~\ref{apdx:waiting_AoII}.
%\end{proof}

On the other hand, for $\mu<\alpha$, Proposition~\ref{prop:threshold} states that for a fixed $r$ there exist a corresponding threshold $n^*_{r,\lambda}$, such that for all states $S=(\delta,r)$ with $\delta\geq n^*_{r,\lambda}\,$ the optimal policy is $y=1$, whereas if $\delta<n^*_{r,\lambda}$ the optimal policy is $y=0$. The next proposition implies that it suffices to define only the threshold for $r=0$, $n^*_{0,\lambda}$.

\begin{proposition}~\label{prop:threshold_r}
Let $n^*_{r,\lambda}$ denote the optimal threshold when the transmission count equals $r$. The sequence $\{n^*_{r,\lambda}\}$ is non-increasing w.r.t. $r$.
\end{proposition}

We elaborate on the consequences of Proposition~\ref{prop:threshold_r}. Suppose that the system has reached the state $S_t = (\delta_t,r_t)$, $\delta_t\!>\!0$, $r_t\!>\!0$. Since $r_t$ is positive, it follows that the previous state was $S_{t-1} = (\delta_t{-}1,r_t{-}1)$ and the optimal action was $y_{t-1}=1$. Therefore, it holds that $\delta_t\!-\!1\!\geq\! n^*_{r_t-1,\lambda}\,$, which implies that $\delta_t\!>\! n^*_{r_t,\lambda}$ due to Proposition~\ref{prop:threshold_r}. It follows that the optimal action for the state $S_t$ is also $y_t\!=\!1$. By induction, we infer that it suffices to find the threshold $n^*_{0,\lambda}$. After the instantaneous AoII reaches $n^*_{0,\lambda}$, the optimal action is to continuously transmit until the AoII becomes zero. %Reversely, the optimal action cannot be to transmit unless the instantaneous AoII is greater or equal to $n^*_{0,\lambda}$.}

Exploiting the previous results, we can derive the following theorem, which describes the solution to the Lagrangian MDP problem~\eqref{eq:lagrange_problem}.

\begin{theorem}\label{thm:optimal_policy}
The optimal threshold $n^*_{0,\lambda}$ is equal to

\begin{equation}\label{eq:optimal_thresh}
\begin{aligned}
    n^*_{0,\lambda} = \min\{&n_0\in\mathbbm{N^*} :\\
    &(1\!-\!\mu) V(n_0{+}1,0)\!-\!V(n_0,0)\!+\!f(n_0)\!-\!g_{n_0}\!>\!0\},
\end{aligned}
\end{equation}
where $g_{n_0}$ is the average cost achieved by the threshold $n_0$. The values of $V(n_0,0)$, $V(n_0{+}1,0)$ and $g_{n_0}$ are computed via the expressions in Table~\ref{tab:opt_policy_analytic}.

\begin{table*}
\normalsize
\hrule
~

    \begin{equation*}
        V(\delta,0) = \sum_{l=0}^{\infty}\left[\left(f(\delta{+}l)+ \lambda -g_{n_0}\right)\sum_{i=0}^{\infty}P_{0,i}^l\right],\quad \text{for $\delta\geq n_0$}\,.
    \end{equation*}
    ~

    \begin{equation*}
        g_{n_0} = \frac{\frac{f(0)}{1-\alpha}+\sum\limits_{i=0}^{n_0-2} (1-\mu)^i f(i{+}1)+(1-\mu)^{n_0-1}\sum\limits_{l=0}^{\infty}\left[\left(f(n_0{+}l)+ \lambda \right)\sum\limits_{i=0}^{\infty}P_{0,i}^l\right]}{\frac{1}{1-\alpha}+\sum\limits_{i=0}^{n_0-2}(1-\mu)^i+(1-\mu)^{n_0-1}\sum\limits_{l=0}^{\infty}\sum\limits_{i=0}^{\infty}P_{0,i}^l},
    \end{equation*}
    
    where
    \begin{equation*}
            P\triangleq
    \bordermatrix{%
    &\scriptstyle\uline{0}&\scriptstyle\uline{1}&\scriptstyle\uline{2}&\scriptstyle\uline{3}&\scriptstyle\uline{4}&\dots&\scriptstyle\uline{j}\cr%&\scriptstyle\uline{j+1}&\scriptstyle\uline{j+2}&\dots\ \cr
    \hfill\scriptstyle\uline{0}\hfill&\gamma_2(0) &\gamma_1(0) &0 &0 &0 &\dots &0\cr%&0 &0 &\dots\cr
    \hfill\scriptstyle\uline{1}\hfill&\gamma_2(1) &0 &\gamma_1(1) &0 &0 &\dots &0\cr%&0 &0 &\dots\cr
    \hfill\scriptstyle\uline{2}\hfill&\gamma_2(2) &0 &0 &\gamma_1(2) &0 &\dots &0\cr%&0 &0 &\dots\cr
    \hfill\scriptstyle\vdots\hfill&\vdots &  &  &  &\ddots & &\vdots\cr% & & &\dots\cr
    \hfill\scriptstyle\vdots\hfill&\vdots &  &  &   & &\ddots &\vdots\cr% & & &\dots\cr
    \hfill\scriptstyle\uline{j-1}\hfill&\gamma_2(j-1) &0 &0 &0 &0 &\dots &\gamma_1(j-1)\cr% & & &\dots\cr
    \hfill\scriptstyle\uline{j}\hfill&\gamma_2(j) &0 &0 &0 &0 &\dots &0\cr% &\gamma_1(j) &0 &\dots\cr
    %\hfill\scriptstyle\vdots\hfill&\vdots &  & & &\vdots  & & &\ddots &\vdots\cr
    },\quad j\rightarrow\infty\,.
    \end{equation*}

\hrule
\caption{The analytic expressions of $V(\delta,0)$ for $\delta\geq n_0$ and $g_{n_0}$.}
\label{tab:opt_policy_analytic}
\end{table*}

\end{theorem}

\section{Optimal Policy For The Constrained Problem}\label{sec:CMDP}

Up to this point, we have studied the Lagrangian MDP~\eqref{eq:lagrange_problem}, with our final objective being the solution of the CMDP~\eqref{eq:opt_problem}. To this end, we need to track the transmission rate achieved by the solution of the Lagrangian MDP.

Let $\pi_\lambda$ denote the policy that solves the Lagrangian MDP with parameter $\lambda$ and let $C_{\pi_\lambda}$ denote the transmission rate that is achieved from this policy.

\begin{proposition}\label{prop:rate}
The achieved transmission rate of the threshold-based policy with threshold $n_0$ is equal to

\begin{equation}\label{eq:achiev_rate_result}
    C_{\pi_\lambda} = q_{0,0} (1-\alpha)(1-\mu)^{n_0-1} \sum_{0\leq h}\, \sum_{0\leq r\leq h} m(h,r),
\end{equation}
where the function $m(\cdot,\cdot)$ is recursively defined by
\begin{align}
    &m(h,r) = \begin{cases} m(h{-}r,0) \prod_{j=0}^{r-1}\gamma_1(j)& \text{if\, $0\leq h,\ 0\leq r\leq h$}\,,\nonumber \\
                             0& \text{if\, $0\leq h,\ h < r$}\,,
                \end{cases}\\
    &m(h,0) = \sum_{k=0}^{h-1}\left[\gamma_2(k)m(h{-}k{-}1,0) \prod_{j=0}^{k-1}\gamma_1(j)\right], \quad 1 \leq h\,,\nonumber\\
    &m(0,0) = 1\,,\\
            &\text{with the convention that $\prod\limits_{k=\phi}^{\chi}(\cdot)=1$ if $\phi>\chi$\ ,}\nonumber
\end{align}
and $q_{0,0}$ is the stationary probability of the process being at state $(0,0)$, which is equal to

\begin{equation}\label{eq:q00}
\begin{aligned}
    {q_{0,0} = \Biggl(1+(1-\alpha)\Biggl(}& \sum_{k=1}^{n_0-1}(1-\mu)^{k-1}\\
                &+ (1-\mu)^{n_0-1}\sum_{h=0}^\infty \sum_{r=0}^{h} m(h,r) \Biggr)\Biggr)^{-1}.
\end{aligned}
\end{equation}
\end{proposition}

\begin{proof}
The proof is given in the Appendix.
\end{proof}

The main result of this section is summarized with the following theorem.
\begin{theorem}\label{thm:cmdp_policy}
There exists an optimal policy $\pi^*$ of the CMDP, which is a randomized mixture of two stationary threshold-based policies of the Lagrangian MDP $\pi_{\lambda^*}$ and $\pi_{\lambda^-}$, that correspond to the thresholds $n^*_{0,\lambda^*}$ and $n^*_{0,\lambda^-}\,$, respectively. In particular,
\begin{equation}\label{eq:lambda_star}
        \lambda^* \triangleq \inf\{\lambda\in\mathbbm{R^+}:C^*_{\pi_{\lambda}}\leq R\}
\end{equation}
and
\begin{equation}
    n^*_{0,\lambda^-} = n^*_{0,\lambda^*} - 1\,.
\end{equation}
The optimal policy is defined as

\begin{equation}\label{eq:optimal_policy_thm}
    \pi^* = \rho \pi_{\lambda^*} + (1-\rho)\pi_{\lambda^-}\,,
\end{equation}
where $\rho\in(0,1]$ denotes the randomized mixture component. The mixture component is chosen such that the randomized policy has an average transmission rate equal to $R$, and it is explicitly defined by
\begin{equation}
    \rho \triangleq \frac{R - C_{\pi_{\lambda^*}}}{C_{\pi_{\lambda^-}} - C_{\pi_{\lambda^*}}}\,,
\end{equation}
where $C_{\pi_{\lambda^-}}$ and $C_{\pi_{\lambda^*}}$ are the average transmission rates when the thresholds are $n^*_{0,\lambda^-}$ and $n^*_{0,\lambda^*}$, respectively.
\end{theorem}

\begin{proof}
The proof is given in the Appendix.
\end{proof}

Equation~\eqref{eq:optimal_policy_thm} is interpreted such as, at every time slot, the policy $\pi^*$ chooses randomly either $n^*_{0,\lambda^*}$ or $n^*_{0,\lambda^-}$ as a threshold, with probability $\rho$ and $1-\rho$, respectively.

Finally, the following proposition will be exploited for the algorithmic utilization of Theorem~\ref{thm:cmdp_policy}, as it will be analyzed in Section~\ref{sec:algorithm}.

\begin{proposition}\label{prop:rate_monotonic}
The Lagrange-optimal threshold $n^*_{r,\lambda}$ is non-decreasing with $\lambda$, while the transmission rate $C_{\pi_\lambda}$ is non-increasing with $\lambda$.
\end{proposition}
\begin{proof}
The proof is detailed in the Appendix.
\end{proof}

\section{Proofs of the Structural Properties of the Lagrangian MDP}~\label{sec:Lagrangian_proofs}

In this section, we prove the results on the structural properties of the Lagrangian MDP written in Sec.~\ref{sec:Lagrangian}.

\subsection{Proof of Lemma~\ref{lem:value_delta}}\label{apdx:value_delta}

    From \eqref{eq:transition_y0}, \eqref{eq:transition_y1} and \eqref{eq:bellman_def} we obtain the following Bellman equations for our problem,
    \begin{align}
        &g + V(0,0) = \min\{f(0) + \alpha V(0,0)+(1-\alpha)V(1,0),\nonumber\\
            &\quad f(0) + \lambda + \alpha V(0,0) + (1-\alpha)V(1,0)\},\label{eq:bellman_00}\\
        &g + V(\delta,r) = \min\{f(\delta)+(1-\mu) V(\delta{+}1,0)+\mu V(0,0),\nonumber\\
        &\quad f(\delta)+\lambda + \gamma_1(r) V(\delta{+}1,r{+}1)+\gamma_2(r)V(\delta{+}1,0)\nonumber\\
        &\quad \!+\left(1-\gamma_1(r)-\gamma_2(r)\right)V(0,0)\},\ \text{for $\ \delta>0,\ r\geq 0$}\,.\label{eq:bellman_dr}
    \end{align}
    Notice that the first part of the minimum operators in \eqref{eq:bellman_00} and \eqref{eq:bellman_dr} corresponds to $y=0$, whereas the second part corresponds to $y=1$.
    
    The relative value iteration (RVI) algorithm is employed next, which approximates the value function $V(\cdot,\cdot)$ in an iterative fashion and, in particular, if it is convergent, it finds the true value function. Let $V_t(\cdot,\cdot)$ denote the estimation of the value function at iteration $t\in\mathbbm{N}$. Next, define the (exact) Bellman operator,
    
    \begin{equation}
        TV_t(S) \triangleq \min_{y}\{f(\delta)+\lambda y + \sum_{S'}P(S'\mid S, y)V_t(S')\}.
    \end{equation}
    
    Let $S_0=(0,0)$ be the arbitrary reference state the RVI is anchored to. Without loss of generality, we may assume that $V_0(S) = f(\delta),\  \forall\ S=(\delta,r)\in\mathcal{S}$. Then, the RVI updates its estimate as follows:
    
    \begin{definition}[Relative Value Iteration]
    \begin{equation}\label{eq:rvi_update}
        V_{t+1}(S) = TV_t(S) - TV_t(0,0),\quad \forall\ S\in\mathcal{S},\ t\in\mathbbm{N}\,.
    \end{equation}
    \end{definition}
    
    By setting $S=(0,0)$ in \eqref{eq:rvi_update}, we deduce that $V_t(0,0) = 0,\ \forall\ t\in\mathbbm{N^*}$. Since $V_t(\cdot,\cdot)$ converges to the true $V(\cdot,\cdot)$, this means that 
    \begin{equation}\label{eq:V00}
        V(0,0) = 0\,.
    \end{equation}
    Plugging~\eqref{eq:V00} in~\eqref{eq:bellman_00}, the min operator is solved immediately as the first argument is always smaller than the second. This yields the following useful remark.
    
    \begin{remark}\label{rem:opt_action_0}
        The optimal action for $S=(0,0)$ is to wait ($y=0$).
    \end{remark}
    
    Furthermore, plugging~\eqref{eq:V00} in both~\eqref{eq:bellman_00} and~\eqref{eq:bellman_dr} yields the following Bellman equations, which will be used in the rest of our proofs:
    
    \begin{align}
        &g = f(0) + (1-\alpha)V(1,0)\ ,\label{eq:bellman_00_simple}\\
        &g + V(\delta,r) = \min\{f(\delta)+(1-\mu) V(\delta{+}1,0),\nonumber\\
        &\quad f(\delta)+\lambda + \gamma_1(r) V(\delta{+}1,r{+}1)+\gamma_2(r)V(\delta{+}1,0)\},\label{eq:bellman_dr_simple}\\
                        &\omit\hfill\text{for $\ \delta>0,\ r\geq 0\,.$}\nonumber
    \end{align}

    Our goal is to show that $V_t(\delta^+,r)>V_t(\delta^-,r)\ \forall\ \delta^+>\delta^-,\ r\in\mathbbm{N},\ t\in\mathbbm{N}$. Suppose that $V_t(\delta^+,r)> V_t(\delta^-,r)$. The condition holds for $t=0$, since $V_0(\delta,r) = f(\delta)$ and $f(\cdot)$ is an increasing function. Suppose that the condition holds up to some iteration $t$. Then, we write the Bellman operators for the two states:
    
    \begin{align}
        &TV_t(\delta^+,r) = \min\{f(\delta^+)+(1-\mu) V_t(\delta^+ {+}1,0),\nonumber\\
        &\quad f(\delta^+)+\lambda + \gamma_1(r) V_t(\delta^+ {+}1,r{+}1)+\gamma_2(r)V_t(\delta^+ {+}1,0)\},\\
        &TV_t(\delta^-,r) = \min\{f(\delta^-)+(1-\mu) V_t(\delta^- {+}1,0),\nonumber\\
        &\quad f(\delta^-)+\lambda + \gamma_1(r) V_t(\delta^- {+}1,r{+}1)+\gamma_2(r)V_t(\delta^- {+}1,0)\}.
    \end{align}
    
    By assumption, we have that $V_t(\delta^+,r)\!>\!V_t(\delta^-,r)\ \forall\ \delta^+\!>\!\delta^-\,,\ r\in\mathbbm{N}$, while $f(\delta^+)\!>\!f(\delta^-)$ due to the monotonicity of $f(\cdot)$. From these two inequalities we deduce that $TV_t(\delta^+,r)\!>\!TV_t(\delta^-,r)$. Using this result in \eqref{eq:rvi_update}, it yields that $V_{t+1}(\delta^+,r)\!>\!V_{t+1}(\delta^-,r)$. By induction, we conclude that $V_{t}(\delta^+,r) \!>\! V_{t}(\delta^-,r)\ \forall\ t\in\mathbbm{N}$. Thus, the function $V(\delta,r)$ is increasing w.r.t $\delta$.

\subsection{Proof of Lemma~\ref{lem:value_r}}\label{apdx:value_r}

    In the following, we prove the result for $\mu<\alpha$. The case $\mu\geq\alpha$ can be proved similarly with the same approach. For the proof, we will invoke directly the Bellman equations~\eqref{eq:bellman_def}. Our goal is to show that $V(\delta,r^+)\leq V(\delta,r^-), \forall\ r^+\geq r^-$. When $\delta=0$, we always have $r^+=r^-=0$ and the inequality holds trivially. Therefore, we will examine the states where $\delta>0$.
    
    Let $S_0^-\triangleq(\delta,r^-)$ and $S_0^+\triangleq(\delta,r^+)$. Assume a sequence of actions $\{y_t^-\}$, corresponding to the optimal policy starting from state $S_0^-$. Let $\{S_t^-\}$ and $\{S_t^+\}$ be the sequence of states after following the actions $\{y_t^-\}$, starting from the states $S_0^-$ and $S_0^+$, respectively. Clearly, $\{y_t^-\}$ is a sub-optimal policy for $\{S_t^+\}$. Therefore,
    
    \begin{equation}\label{eq:V_plus_minus_ineq}
        V(S_0^+)\leq f(\delta) + \lambda y_0^- -g +\sum_{S^+_1}P(S^+_1\mid S^+_0, y_0^-)V(S^+_1).
    \end{equation}
    
    Suppose that $y_0^-=0$. Then, the action-conditional transition probabilities of $S^+_0$ and $S^-_0$ (ref.~\eqref{eq:transition_y0}) are the same and thus,
    
    \begin{equation}
    \begin{aligned}
        V(S_0^+)\leq& f(\delta) + \lambda y_0^- -g +\sum_{S^-_1}P(S^-_1\mid S^-_0, y_0^-)V(S^-_1)\\
                =&V(S_0^-).
    \end{aligned}
    \end{equation}

    Instead, suppose that a sequence of $k$ transmissions follows before the first wait action, i.e. $y_0^-\!=\!y_1^-\!=\!\dots\!=\!y_{k-1}^-\!=\!1$ and $y_k^-\!=\!0$. From~\eqref{eq:bellman_dr_simple}, we have the following relations for $V(\delta,r^+)$.

    \begin{equation}\label{eq:V_y1}
    \begin{aligned}
    V(\delta,r^+) = &f(\delta)+\lambda-g+ \gamma_2(r^+)V(\delta{+}1,0) \\
                    &+\gamma_1(r^+) V(\delta{+}1,r^+{+}1), \qquad\quad\quad\text{if $y=1$}\,.
    \end{aligned}
    \end{equation}

    The recursive relation~\eqref{eq:V_y1} expresses the value function $V(\delta,r^+)$ as \mbox{$(f(\delta)\!+\!\lambda\!-\!g)$}, plus the average value of the next state, which is $V(\delta{+}1,0)$ with probability $\gamma_2(r^+)$, and $V(\delta{+}1,r^+{+}1)$ with probability $\gamma_1(r^+)$. Therefore, the possible states of the next time slot that contribute to the value function of the state $(\delta,r^+)$ have a deterministic AoII variable equal to $\delta+1$ and a stochastic transmission count variable. We can exploit the Markovian nature of the transitions to represent their probability as a transition matrix. Define the ($1$-step) transition matrix,

    \begin{equation}
    P\!\triangleq\!
    \let\quad\thinspace{\bordermatrix{%
    &\scriptstyle\uline{0}&\scriptstyle\uline{1}&\scriptstyle\uline{2}&\scriptstyle\uline{3}&\dots&\scriptstyle\uline{r^+{+}k}\cr
    \hfill\scriptstyle\uline{0}\hfill&\gamma_2(0) &\gamma_1(0) &0 &0 &\dots &0\cr
    \hfill\scriptstyle\uline{1}\hfill&\gamma_2(1) &0 &\gamma_1(1) &0 &\dots &0\cr
    \hfill\scriptstyle\uline{2}\hfill&\gamma_2(2) &0 &0 &\gamma_1(2) &\dots &0\cr
    \hfill\scriptstyle\vdots\hfill&\vdots &  &  &  &\ddots &\cr
    \hfill\!\scriptstyle\uline{r^+{+}k{-}1}\!\!\!\hfill&\!\!\gamma_2(r^+{+}k{-}1)\!\! &0 &0 &0 &\dots &\gamma_1(r^+{+}k{-}1)\!\!\!\cr
    \hfill\scriptstyle\uline{r^+{+}k}\hfill&\gamma_2(r^+{+}k) &0 &0 &0 &\dots &0
    }}
    \end{equation}
    Thus, $P$ is the $1$-step action-dependent ($y\!=\!1$) probability matrix\footnote{$P$ is not a stochastic matrix, since the summation of its rows is lower than $1$. This is true since $P$ lacks the transitions to the state $(0,0)$. However, we do not lose any information since the value function of $(0,0)$ is zero.} for the transitions $(\delta,r^+)\!\rightarrow\!(\delta{+}1,0)$ and $(\delta,r^+)\!\rightarrow\!(\delta{+}1,r^+{+}1)$. As shown above, we index the rows and columns of $P$ from $0$ to $r^+{+}k$. Notice that the matrix is big enough to express the transitions for the $k$ sequential transitions starting from the state $(\delta,r^+)$. The $l$-th step transition probabilities are given by $P^l$. That is, the $l$-th step transition $(\delta,r^+)\!\overset{l}{\rightarrow}\!(\delta{+}l,i)$ has probability equal to $P^l_{r^+,i}$\,.
    
    With the definition of the matrix $P$ and the above observations at hand, we can expand the recursions in the RHS of~\eqref{eq:V_plus_minus_ineq} for the first $k$ steps, as shown in Table~\ref{tab:lemma3_proof_eq1}. The second line follows from the fact that the transition probabilities of $(\delta{+}k,i)$ are the same for every $i$ when $y_k^-\!=\!0$~\eqref{eq:transition_y0}. The fourth line follows from the fact that $\sum_{i}P_{r,i}^{s}$ is non-increasing with $r$ for any power $s$ of the matrix $P$, which we prove next.

    \begin{table*}
    \normalsize
    \hrule
    ~
    
    \begin{equation*}
    \begin{aligned}
        V(S_0^+)\leq &\sum_{l=0}^{k}\left[\left(f(\delta+l)+ \lambda y_l^- -g\right)\sum_{i=0}^{r^+{+}k}P_{r^+,i}^l\right] \!+\!\sum_{i=0}^{r^+{+}k}\left[P_{r^+,i}^{k}\sum_{S^+_{k+1}}\!\left[P\left((\delta{+}k{+}1,r^+_{k+1})\!\mid\! (\delta{+}k,i), y_k^-\right)V(\delta{+}k{+}1,r^+_{k+1})\right]\right]\\
        =&\sum_{l=0}^{k}\left[\left(f(\delta+l)+ \lambda y_l^- -g\right)\sum_{i=0}^{r^+{+}k}P_{r^+,i}^l\right] \!+\!\sum_{i=0}^{r^+{+}k}\left[P_{r^+,i}^{k}\sum_{S^-_{k+1}}\!\left[P\left((\delta{+}k{+}1,r^-_{k+1})\!\mid\! (\delta{+}k,r^-_{k}), y_k^-\right)V(\delta{+}k{+}1,r^-_{k+1})\right]\right]\\
        =&\sum_{l=0}^{k}\left[\left(f(\delta+l)+ \lambda y_l^- -g\right)\sum_{i=0}^{r^+{+}k}P_{r^+,i}^l\right] \!+\!\sum_{i=0}^{r^+{+}k}\left[P_{r^+,i}^{k}\right]\sum_{S^-_{k+1}}\!\left[P\left((\delta{+}k{+}1,r^-_{k+1})\!\mid\! (\delta{+}k,r^-_{k}), y_k^-\right)V(\delta{+}k{+}1,r^-_{k+1})\right]\\
        \overset{!}{\leq}&\sum_{l=0}^{k}\left[\left(f(\delta+l)+ \lambda y_l^- -g\right)\sum_{i=0}^{r^-{+}k}P_{r^-,i}^l\right] \!+\!\sum_{i=0}^{r^-{+}k}\left[P_{r^-,i}^{k}\right]\sum_{S^-_{k+1}}\!\left[P\left((\delta{+}k{+}1,r^-_{k+1})\!\mid\! (\delta{+}k,r^-_{k}), y_k^-\right)V(\delta{+}k{+}1,r^-_{k+1})\right]\\
        = & V(S_0^-)
    \end{aligned}
    \end{equation*}
    
    \hrule
    \caption{Steps for proving the inequality $V(S_0^+)\leq V(S_0^-)$.}
    \label{tab:lemma3_proof_eq1}
    \end{table*}

    We will use induction to prove that $\sum_{i}P_{r,i}^{s}$ is non-increasing with $r$ for any power $s$ of the matrix $P$. We know that the property holds for $s=0$ and $s=1$, since $\sum_{i}P_{r,i}^{0}=1$ and $\sum_{i}P_{r,i}^{1} = \gamma_2(r)+\gamma_1(r)$, which is non-increasing when $\mu<\alpha$. Suppose that the property holds up to some power $h$. We will prove that it must also hold for the power $h\!+\!1$. Let $P_{r,:}^{h+1}$ denote the $r$-th row vector of $P^{h+1}$. We can write $P_{r,:}^{h+1}$ as the following product,

    \begin{equation}
        P_{r,:}^{h+1} = P_{r,:} \cdot P^{h}\,.
    \end{equation}
    Expanding the product and summing all the elements, we derive the following relation
    \begin{equation}\label{eq:P_recursive_relation1}
        \sum_{i}P_{r,i}^{h+1} = \gamma_2(r)\sum_{i}P_{0,i}^{h} + \gamma_1(r)\sum_{i}P_{r+1,i}^{h}\,.
    \end{equation}
    After replacing the $\gamma_1(r)$ and $\gamma_2(r)$ with their definitions~\eqref{eq:transition_y0},~\eqref{eq:transition_y1} and some minor algebraic manipulations,~\eqref{eq:P_recursive_relation1} becomes
    \begin{equation}\label{eq:P_recursive_relation2}
    \begin{aligned}
        \sum_{i}P_{r,i}^{h+1} = &(1-\alpha)\sum_{i}P_{0,i}^{h} \\
        &+ \left(1-p(r)\right)\left(\alpha\sum_{i}P_{r+1,i}^{h}-\mu\sum_{i}P_{0,i}^{h}\right).
    \end{aligned}
    \end{equation}
    
    The first term in~\eqref{eq:P_recursive_relation2} is constant w.r.t. $r$. The second term is a product of two non-increasing functions of $r$. In particular, $\left(1-p(r)\right)$ is non-increasing since $p(r)$ is non-decreasing, and $\left(\alpha\sum_{i}P_{r+1,i}^{h}\!-\!\mu\sum_{i}P_{0,i}^{h}\right)$ is non-increasing since $\sum_{i}P_{r+1,i}^{h}$ is non-increasing by the induction hypothesis. Therefore, to prove that $\sum_{i}P_{r,i}^{h+1}$ is non-increasing with $r$, it suffices to show that $\left(\alpha\sum_{i}P_{r+1,i}^{h}\!-\!\mu\sum_{i}P_{0,i}^{h}\right)$ is always non-negative. It is easy to verify that, for any row $u$, $\sum_{i}P_{u,i}^{h}$ is minimized when $p(x)=1\ \forall\ x\in\mathbbm{N}$ and maximized when $p(x)=0\ \forall\ x\in\mathbbm{N}$. For those extreme cases, we have

    \begin{align}
        \left.\sum_{i}P_{u,i}^{h}\right\rvert_{p(x)=1\ \forall\ x\in\mathbbm{N}}=&\sum_{j=0}^h(1-\alpha)^j = \frac{1-(1-\alpha)^{h+1}}{\alpha}\\
        \left.\sum_{i}P_{u,i}^{h}\right\rvert_{p(x)=0\ \forall\ x\in\mathbbm{N}}=&\sum_{j=0}^h(1-\mu)^j = \frac{1-(1-\mu)^{h+1}}{\mu}
    \end{align}
    Hence,
    \begin{equation}
    \begin{aligned}
        \alpha\sum_{i}P_{r+1,i}^{h}\!-\!\mu\sum_{i}P_{0,i}^{h}\geq &\alpha\left.\sum_{i}P_{r+1,i}^{h}\right\rvert_{p(x)=1\ \forall\ x\in\mathbbm{N}}\\
        &-\mu \left.\sum_{i}P_{0,i}^{h}\right\rvert_{p(x)=0\ \forall\ x\in\mathbbm{N}}\\
        =&-(1-\alpha)^{h+1}+(1-\mu)^{h+1}\\
        >&0\,.
    \end{aligned}
    \end{equation}

    Thus, $\sum_{i}P_{r,i}^{h+1}$ is non-increasing with $r$, and by induction we infer that $\sum_{i}P_{r,i}^{s}$ is non-increasing with $r$ for any power $s$ of the matrix $P$. This concludes our proof.

\subsection{Proof of Proposition~\ref{prop:threshold}}\label{apdx:threshold}

    First, we prove that the optimal policy is threshold-based under the assumption that $\mu<\alpha$. The following lemma will be used to this end.
    
    \begin{lemma}\label{lem:partial}
    Let a two-variable function $f(x,y):\mathcal{X}\times \mathcal{Y}\mapsto \mathbbm{R}$ that is increasing w.r.t. $x$ and non-increasing w.r.t. $y$. That is, $\frac{\partial f(x,y)}{\partial x}>0$ and $\frac{\partial f(x,y)}{\partial y}\leq 0$. Then, the following inequality holds,
    
    \begin{equation}
        \frac{\partial f(x,y)}{\partial x}\geq \frac{\partial f(x,y+h)}{\partial x}\quad \forall\ x\in\mathcal{X},\ (y+h)\in \mathcal{Y},\ h\geq 0\,.
    \end{equation}
    \end{lemma}
    \begin{proof}[Proof of Lemma~\ref{lem:partial}]
    Assume that $\frac{\partial f(x,y)}{\partial x}\!<\! \frac{\partial f(x,y{+}h)}{\partial x}$. Then, there exists some $x^* \in \mathcal{X}$ s.t. $f(x,y)\!<\!f(x,y{+}h),\ \forall\ x\!\geq\! x^*$. However, this contradicts the monotonicity of $f(x,y)$. Consequently, the assumption is wrong.
    \end{proof}
    
    We already showed that for $\delta=0$, the optimal policy is $y=0$ (Remark~\ref{rem:opt_action_0}). Thus we consider the states where $\delta>0$. Define the action-dependent value functions for $y=1$ and $y=0$, respectively, from~\eqref{eq:bellman_dr_simple}:
    
    \begin{align}
        &V^1(\delta,r) \triangleq -g + f(\delta)+\lambda + \gamma_1(r) V(\delta{+}1,r{+}1)\nonumber\\
                        &\phantom{V^1(\delta,r) \triangleq}+\gamma_2(r)V(\delta{+}1,0),\\
         &V^0(\delta,r)\triangleq -g + f(\delta)+(1-\mu) V(\delta{+}1,0).
    \end{align}
    In addition, define the difference,
    
    \begin{equation}\label{eq:delta}
    \begin{aligned}
        \Delta V(S) &\triangleq V^1(S)-V^0(S) \\
                    &= \lambda \!+\! \gamma_1(r)V(\delta{+}1,r{+}1)\!+\!(\gamma_2(r)\!-\!1\!+\!\mu)V(\delta{+}1,0)\\
                    &= \lambda + (\alpha-\alpha p(r))V(\delta{+}1,r{+}1)\\
                    &\phantom{= \lambda }-(\alpha-\mu p(r))V(\delta{+}1,0).
    \end{aligned}
    \end{equation}
    
    The partial derivative of $\Delta V(S)$ w.r.t. $\delta$ is equal to
    
    \begin{equation}\label{eq:partial_delta}
    \begin{aligned}
        \frac{\partial \Delta V(S)}{\partial \delta} &= (\alpha-\alpha p(r))\frac{\partial V(\delta{+}1, r{+}1)}{\partial \delta}\\
        &\phantom{=}-(\alpha-\mu p(r))\frac{\partial V(\delta{+}1,0)}{\partial \delta}.   
    \end{aligned}
    \end{equation}
    
    At this point, notice that~\eqref{eq:partial_delta} is the difference of two positive terms, since $p(r)\in (0,1)$ and $\mu<\alpha$. Lemma~\ref{lem:partial} will be employed to determine its sign. Specifically, by Lemma~\ref{lem:value_r} and Lemma~\ref{lem:partial} we have that
    
    \begin{equation}\label{eq:auxil_partial_1}
        \frac{\partial V(\delta{+}1, r{+}1)}{\partial \delta} \leq \frac{\partial V(\delta{+}1,0)}{\partial \delta}\,.
    \end{equation}
    Furthermore, since $\mu<\alpha$,
    
    \begin{equation}\label{eq:auxil_partial_2}
        \alpha-\alpha p(r)<\alpha-\mu p(r). 
    \end{equation}
    \noindent Multiplying~\eqref{eq:auxil_partial_1} and~\eqref{eq:auxil_partial_2} by parts, we get
    
    \begin{equation}
        \frac{\partial \Delta V(S)}{\partial \delta} < 0\ .
    \end{equation}
    
    We conclude that $V^0(\delta,r)$ increases at a higher rate than $V^1(\delta,r)$ does. Therefore, as $\delta$ increases, it is possible that $V^0(\delta,r)$ becomes larger than $V^1(\delta,r)$. Let this point be $n^*_{r,\lambda}\,$, where the subscripts $r$ and $\lambda$ indicate that it depends on the transmission count $r$ and the Lagrangian parameter $\lambda$. Then, for all $\delta\geq n^*_{r,\lambda}\,$, it holds that $V^0(\delta,r)>V^1(\delta,r)$, which means that the optimal policy is $y=1$, while for all $\delta<n^*_{r,\lambda}$ it is $y=0$. Therefore, when $\mu<\alpha$, the optimal policy is threshold-based w.r.t. $\delta$.
    
    Next, we prove that the optimal policy is to always wait under the assumption that $\mu\geq\alpha$. In this case, we have that
    
    \begin{equation}\label{eq:auxil_partial_3}
        \alpha-\alpha p(r)\geq\alpha-\mu p(r),
    \end{equation}
    and by Lemma 3,
    \begin{equation}\label{eq:auxil_partial_4}
        V(\delta{+}1, r{+}1)\geq V(\delta{+}1, 0).
    \end{equation}
    Multiplying~\eqref{eq:auxil_partial_3} and~\eqref{eq:auxil_partial_4} by parts, it follows that the function $\Delta V(S)$ in \eqref{eq:delta} is non-negative, i.e., $V^1(\delta,r)\geq V^0(\delta,r)$. In other words, the optimal policy is always $y=0$.

\subsection{Proof of Proposition~\ref{prop:threshold_r}}\label{apdx:threshold_r}
    
    It suffices to show that the function $\Delta V(\delta,r)=V^1(S)-V^0(S)$ in~\eqref{eq:delta} is non-increasing with $r$. In the proof of Lemma 3 (Sec.~\ref{apdx:value_r}), we essentially showed that the action-dependent value function $V^1(S)$ is non-increasing with $r$. On the other hand, $V^0(S)$ is independent of $r$. Hence, $\Delta V(\delta,r)$ is non-increasing with $r$.

\subsection{Proof of Theorem~\ref{thm:optimal_policy}}\label{apdx:optimal_policy}
    We look for the threshold $n_0\in\mathbbm{N^*}$ such that the optimal action is $y=1$ when $\delta\geq n_0$, and $y=0$ when $\delta< n_0$. We exclude zero from the set of interest due to Remark~\ref{rem:opt_action_0}.
    
    Examining~\eqref{eq:bellman_dr_simple}, for all $\delta\geq n_0$ the second branch of the minimum operator must be smaller than the first branch and thus
    
    \begin{equation}\label{eq:thresh_ineq}
    \begin{aligned}
        &f(\delta)+(1\!-\!\mu) V(\delta{+}1,0)>f(\delta)+\lambda+\gamma_1(r)V(\delta{+}1,r{+}1)\\
         &\phantom{f(\delta)+(1-\mu) V(\delta{+}1,0)>}+\gamma_2(r)V(\delta{+}1,0)\\
        \Rightarrow& (1\!-\!\mu\!-\!\gamma_2(r))V(\delta{+}1,0)\!-\!\gamma_1(r)V(\delta{+}1,r{+}1)> \lambda\,,\\
        &\omit\hfill \text{for $\ \delta\geq n_0$}\,.
    \end{aligned}
    \end{equation}
    
    Besides, from~\eqref{eq:bellman_dr_simple} we know that
    
    \begin{equation}\label{eq:V_dr_recursion}
    \begin{aligned}
        V(\delta,r)=&-g_{n_0}+f(\delta)+\lambda+\gamma_1(r)V(\delta{+}1,r{+}1)\\
                        &+\gamma_2(r)V(\delta{+}1,0)\ ,\quad \text{for $\ \delta\geq n_0$}\,,
    \end{aligned}
    \end{equation}
    where $g_{n_0}$ is the average cost achieved by employing the threshold $n_0$.
    
    Combining~\eqref{eq:thresh_ineq} and~\eqref{eq:V_dr_recursion} and setting $r=0$, we obtain
    
    \begin{equation}\label{eq:ineq_n}
        (1-\mu) V(\delta{+}1,0)-V(\delta,0) +f(\delta)-g_{n_0}>0\,,\quad \text{for $\ \delta\geq n_0$}\,,
    \end{equation}
    which implies that
    
    \begin{equation}
    \begin{aligned}
        n^*_{0,\lambda} = \min\{&n_0\in\mathbbm{N^*}\!:\\
        &(1\!-\!\mu) V(n_0{+}1,0)\!-\!V(n_0,0)\!+\!f(n_0)\!-\!g_{n_0}\!>\!0\},
    \end{aligned}
    \end{equation}
    
    Next, we have to calculate $V(n_0,0)$, $V(n_0{+}1,0)$ and $g_{n_0}$. We can expand the recursive relation in~\eqref{eq:V_dr_recursion} as we did with~\eqref{eq:V_y1} in the proof Lemma 3~(Sec.~\ref{apdx:value_r}). Thus, we derive the expression of $V(\delta,r)$ for $\delta\geq n_0$ in Table~\ref{tab:V_dr_analytic}. Since~\eqref{eq:V_dr_recursion} holds for both $\delta=n_0$ and $\delta=n_0{+}1$, the expression can be used to calculate both $V(n_0,0)$ and $V(n_0{+}1,0)$. Similarly, we have that
    
    \begin{equation}\label{eq:V_dr_recursion<}
            V(\delta,r) = -g_{n_0}+f(\delta)+(1-\mu) V(\delta{+}1,0),\ \text{for $\ 0<\delta<n_0$}\,,
    \end{equation}
    which leads to the expression for $0<\delta<n_0$ in Table~\ref{tab:V_dr_analytic}.
    
    \begin{table*}
    \normalsize
    \hrule
    ~
    
        \begin{equation*}
            V(\delta,r) = \begin{cases}\sum\limits_{l=0}^{\infty}\left[\left(f(\delta+l)+ \lambda -g\right)\sum\limits_{i=0}^{\infty}P_{r,i}^l\right] &\text{if $\delta \geq n_0$}\,,\\
            -g_{n_0} \sum\limits_{i=0}^{n_0-\delta-1} (1-\mu)^i + \sum\limits_{i=0}^{n_0-\delta-i}(1-\mu)^i f(\delta{+}i) +(1-\mu)^{n_0 -\delta}V(n_0,r) &\text{if $0<\delta<n_0$}\,.
            \end{cases}
        \end{equation*}
        
    \hrule
    \caption{The analytic expressions of $V(\delta,r)$.}
    \label{tab:V_dr_analytic}
    \end{table*}
    
    Furthermore, we can rewrite~\eqref{eq:bellman_00_simple} as follows,
    
    \begin{equation}
        g_{n_0} = f(0)+(1-\alpha)V(1,0).\label{eq:g_n_1}
    \end{equation}
    Setting $\delta=1,r=0$ in the expressions of $V(\delta,r)$ in Table~\ref{tab:V_dr_analytic}, we derive
    
    \begin{equation}\label{eq:V_10}
    \begin{aligned}
        V(1,0) &= -g_{n_0}\sum_{i=0}^{n_0-2}(1-\mu)^i\\
                &\phantom{=}+\sum_{i=0}^{n_0-2} (1-\mu)^i f(i{+}1)+(1-\mu)^{n_0-1}V(n_0,0).
    \end{aligned}
    \end{equation}
    To verify that~\eqref{eq:V_10} holds for all values of $n_0$, first notice that if $n_0>1$, then~\eqref{eq:V_10} follows directly from the replacement of $\delta=1$ at the expression of $V(\delta,r)$ for $0\!<\!\delta\!<\!n_0$ in~Table~\ref{tab:V_dr_analytic}. On the other hand, if $n_0\leq 1$, then, by Remark~\ref{rem:opt_action_0} it must hold that $n_0=1$, in which case~\eqref{eq:V_10} simply states that $V(1,0)=V(1,0)$.
    
    Combining~\eqref{eq:g_n_1} and~\eqref{eq:V_10} we get
    
    \begin{equation}\label{eq:g_n_2}
    \begin{aligned}
        g_{n_0} &= f(0) + (1-\alpha)\Bigg(-g_{n_0}\sum_{i=0}^{n_0-2}(1-\mu)^i \\
        &\phantom{=}+ \sum_{i=0}^{n_0-2}(1-\mu)^i f(i{+}1)+(1-\mu)^{n_0-1}V(n_0,0)\Bigg) \\
        \Rightarrow g_{n_0} &= \frac{\frac{f(0)}{1-\alpha}+\sum\limits_{i=0}^{n_0-2} (1{-}\mu)^i f(i{+}1)+(1{-}\mu)^{n_0-1}V(n_0,0)}{\frac{1}{1-\alpha}+\sum\limits_{i=0}^{n_0-2}(1-\mu)^i}\,.
    \end{aligned}
    \end{equation}
    Replacing $V(n_0,0)$ from Table~\ref{tab:opt_policy_analytic}, we end up with the expression for the $g_{n_0}$ written in Table~\ref{tab:opt_policy_analytic}.

\section{Algorithmic Implementation}\label{sec:algorithm}
Up to this point, the optimal transmission policy has been derived on a theoretical basis. This section examines the practical computation of the optimal policy. %Our approach is similar to~\cite{bib:aoii}.

First, we elaborate on the computation of the series in the expressions of $V(n_0,0)$ and $g_{n_0}$ in Table~\ref{tab:opt_policy_analytic} and $q_{0,0}$ in~\eqref{eq:q00}. The series $\sum_{i=0}^{n_0-2}(1-\mu)^i f(i{+}1)$, $\sum_{i=0}^{n_0-2}(1-\mu)^i$ and $\sum_{k=1}^{n_0-1}(1-\mu)^{k-1}$ are finite geometric series and can be evaluated directly; the first via a direct evaluation and the last two via their closed-form expression. The series $\sum_{l=0}^{\infty}\sum_{i=0}^{\infty}P_{0,i}^l$ is a convergent infinite series. More precisely, the sequence $\sum_{i=0}^{\infty}P_{0,i}^l$ converges to zero as $l$ increases, since the elements of $P$ are all lower than $1$, and the summation of every row is also lower than $1$. Moreover, $P_{0,i}^l$ is non-zero only for $i\!\leq\! l\!+\!1$. Consequently, the infinite summation can be approximated by the partial summation $\sum_{l=0}^{l_\epsilon}\sum_{i=0}^{l+1}P_{0,i}^l$, where $l_\epsilon=\argmin_l \{\sum_{i=0}^{l+1}P_{0,i}^l<\epsilon\}$, and $\epsilon$ being a precision constant. The same holds for $\sum_{l=0}^{\infty}\sum_{i=0}^{\infty}\left[P_{0,i}^l\right]f(n_0{+}l)$ due to the assumption in~\eqref{eq:bound_cond}. Similar arguments hold for $\sum_{h=0}^\infty \sum_{r=0}^{h} m(h,r)$.

Having validated the computability of all necessary functions, and due to Proposition~\ref{prop:rate_monotonic}, the algorithmic steps described in~\cite{bib:aoii} are utilized to compute the optimal policy. In the following, we give a summary of the algorithm and refer the reader to~\cite{bib:aoii} for more details.

As a first step, we describe how Theorem~\ref{thm:cmdp_policy} can be used to implement the optimal policy. Consider a method suggesting that some $\lambda$ verifies the inequality $C_{\pi_\lambda}\leq R$ in~\eqref{eq:lambda_star}. To check the validity of the suggestion, we first use Theorem~\ref{thm:optimal_policy} to find the Lagrange-optimal $n_0$ and then $C_{\pi_\lambda}$ is calculated via Proposition~\ref{prop:rate}. This process is straightforward but does not solve the problem of finding the infimum of such $\lambda$'s.

To tackle this issue, we utilize Proposition~\ref{prop:rate_monotonic}. First, we rely on the non-increasing property of $C_{\pi_\lambda}$, which is equivalent to stating that the sequence $\{C_{\pi_\lambda}\}$, for $\lambda\in(0,+\infty)$, is non-increasing. As such, we can use a binary search algorithm to find $\lambda^*$, whose computational complexity is at the order of~$O(\log \lambda^*)$~\cite{bib:algorithms}. The same method can be applied in finding $n^*_{0,\lambda}\,$. In particular, the non-decreasing property of $n^*_{0,\lambda}$ implies that the LHS of the condition inside the minimum operator in~\eqref{eq:optimal_thresh} is non-decreasing with $\lambda$. Therefore, the binary search algorithm can be used to find $n^*_{0,\lambda}$ with complexity $O(\log n^*_{0,\lambda})$.

\section{Numerical Results}\label{sec:results}

In this section, we perform a numerical evaluation of the average AoII under the optimal transmission policy. The results are obtained from simulations of a horizon equal to $T=10^5$. We study the impact of (i) the source dynamics, (ii) the HARQ protocol and (iii) the resource constraint.

Similar to~\cite{bib:aoi_harq}, motivated by previous research on HARQ (ref.~\cite{bib:harq_survey,bib:arq_per}), we model the probability of failed decoding as an exponentially decreasing function, i.e.,

\begin{equation}
    p(r) = 1 - p_e c^r\,,\quad \text{for $\ 0\leq r\leq r_{max}$}\,,
\end{equation}
where $c\in(0,1]$ is the decaying error rate constant, $p_e$ is the packet error rate of the first packet, and $r_{max}$ is the maximum number of allowed retransmissions. In general, $p_e$ and $c$ depend on the channel conditions and the HARQ protocol. For example, a very noisy channel implies a high $p_e$ value. Also, the way that re-transmitted packets are chosen affects $c$. Moreover, for fading channels, $c$ is generally lower in fast-fading channels than in slow-fading (ref.~\cite{bib:harq_fading}).

In the typical case where $r_{max}$ is finite, if the decoder cannot decode with $r_{max}$ packets, then the decoding fails, the packets are discarded and a new round of transmissions begins. We can impose this mechanism with the modulus operator to define the probability function

\begin{equation}
    p_{m}(r) \triangleq 1 - p_e c^{r \bmod (r_{max}{+}1)}\,,\quad \text{for $\ 0\leq r$}\,.
\end{equation}

Additionally, when HARQ is used without soft combining or the standard ARQ is used instead, the previously transmitted packets are not used for decoding. Hence, the decaying constant is $c=1$, or equivalently $r_{max}=0$.

For our experiments, we employ a linear AoII function,

\begin{equation}
    f(\delta_t) = \delta_t\,.
\end{equation}

\begin{figure}
    \centering
    \includegraphics[width=.94\linewidth]{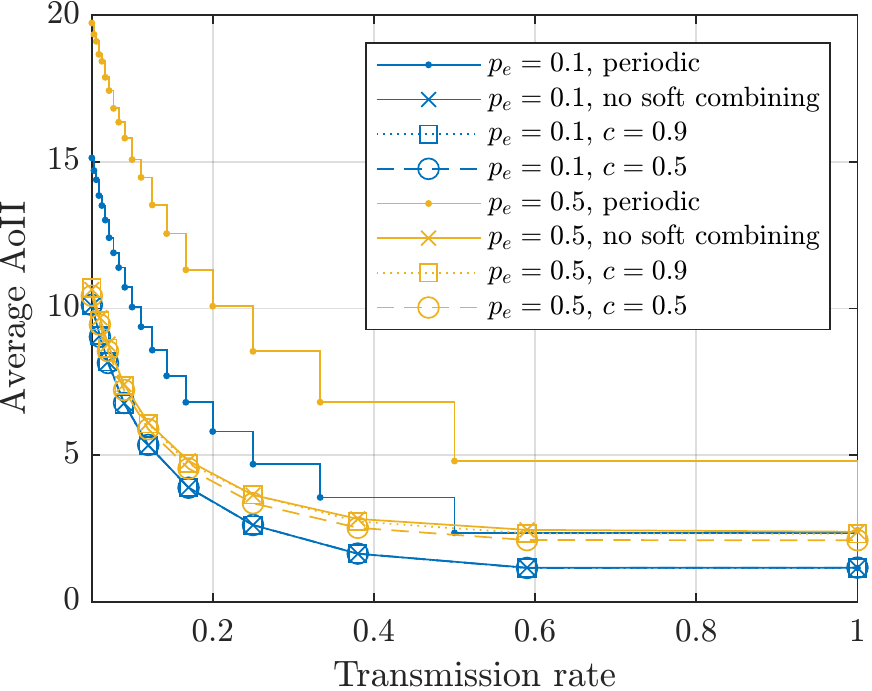}
    \caption{Average AoII versus the transmission rate constraint of the optimal threshold-based policy. The source model parameters are $\alpha=0.5$ and $N=16$ and the maximum retransmission count is $r_{max}=2$.}
    \label{fig:p0_c_N_1}
    ~\\~\\
    \includegraphics[width=.94\linewidth]{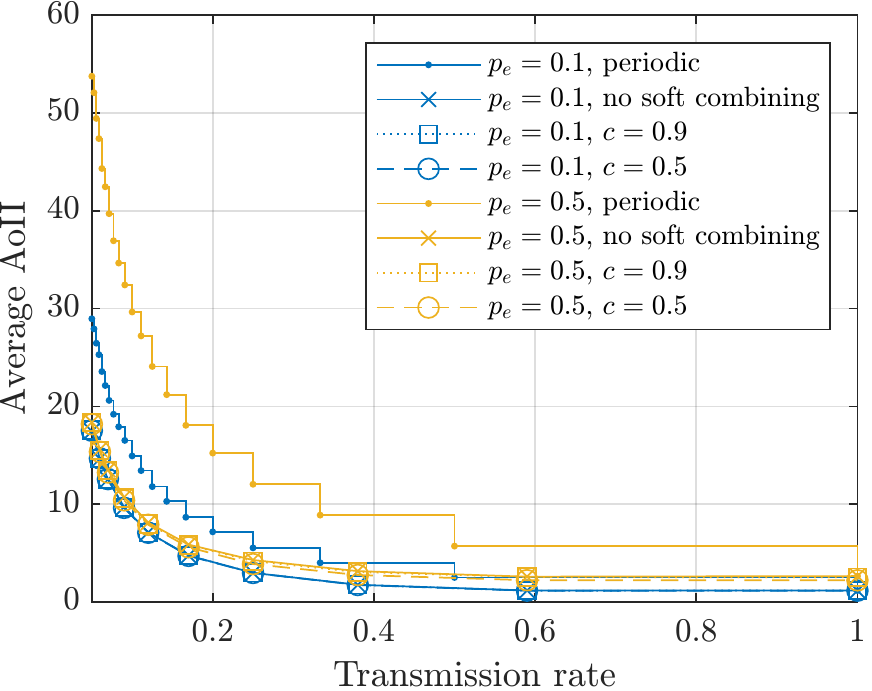}
    \caption{Average AoII versus the transmission rate constraint of the optimal threshold-based policy. The source model parameters are $\alpha=0.5$ and $N=128$ and the maximum retransmission count is $r_{max}=2$.}
    \label{fig:p0_c_N_2}
\end{figure}

\begin{figure}
    \centering
    \includegraphics[width=.94\linewidth]{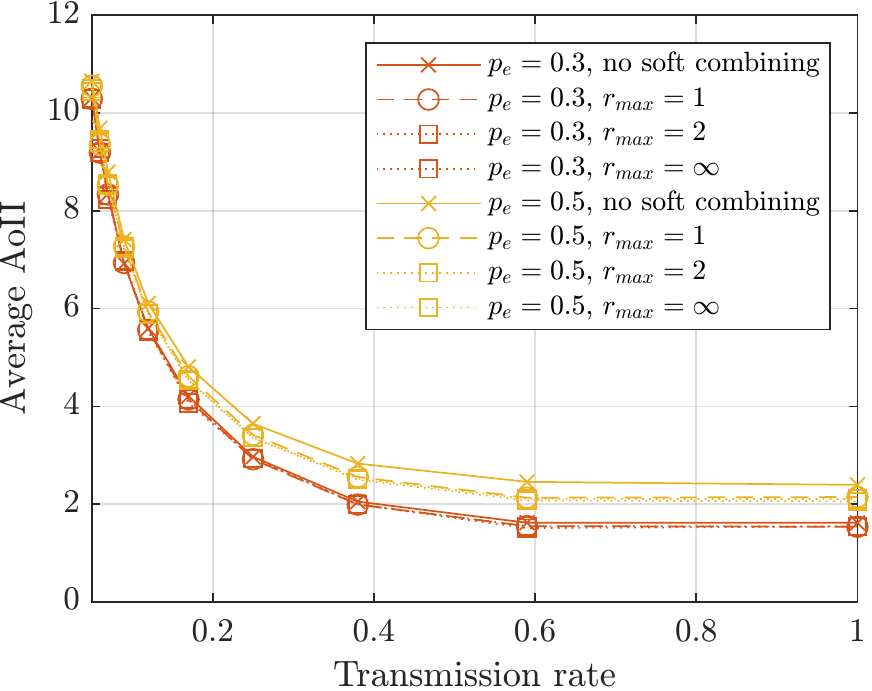}
    \caption{Average AoII versus the transmission rate constraint of the optimal threshold-based policy. The source model parameters are $\alpha=0.5$ and $N=16$ and the HARQ decaying error rate constant is $c=0.5$.}
    \label{fig:p0_r_N_1}
    ~\\~\\
    \includegraphics[width=.94\linewidth]{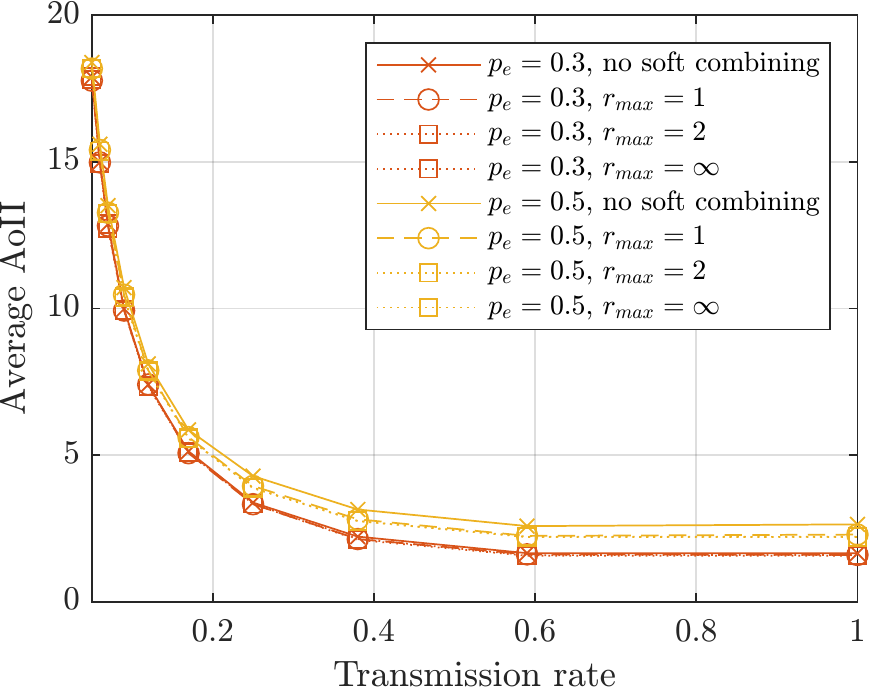}
    \caption{Average AoII versus the transmission rate constraint of the optimal threshold-based policy. The source model parameters are $\alpha=0.5$ and $N=128$ and the HARQ decaying error rate constant is $c=0.5$.}
    \label{fig:p0_r_N_2}
\end{figure}

Our first results concern the case where the optimal policy is threshold-based, i.e. when $\mu\!<\!\alpha$. Figs.~\ref{fig:p0_c_N_1}-\ref{fig:p0_c_N_2} show the average AoII versus the transmission rate for two different source models (same $\alpha$, different $N$) and various $p_e$ and $c$ values, having the maximum number of re-transmissions fixed to $r_{max}\!=\!2$. The scenario where HARQ is without soft combining is also included. As a baseline reference, a deterministic periodic transmission policy that satisfies the resource constraint is included. The periodic policy transmits every $\lceil 1/R \rceil$ time slots. There is a clear advantage of the optimal policy against the periodic one. Additionally, for both sources, there is a notable difference between HARQ with and without soft combining only in the high error rate regime. Besides, there is a large difference in the achieved AoII between the sources with few and many states ($N\!=\!16$ and $N\!=\!128$, respectively) when the allowed transmission rate is small, but they get close as the transmission rate increases.

\begin{figure}
    \centering
    \includegraphics[width=.94\linewidth]{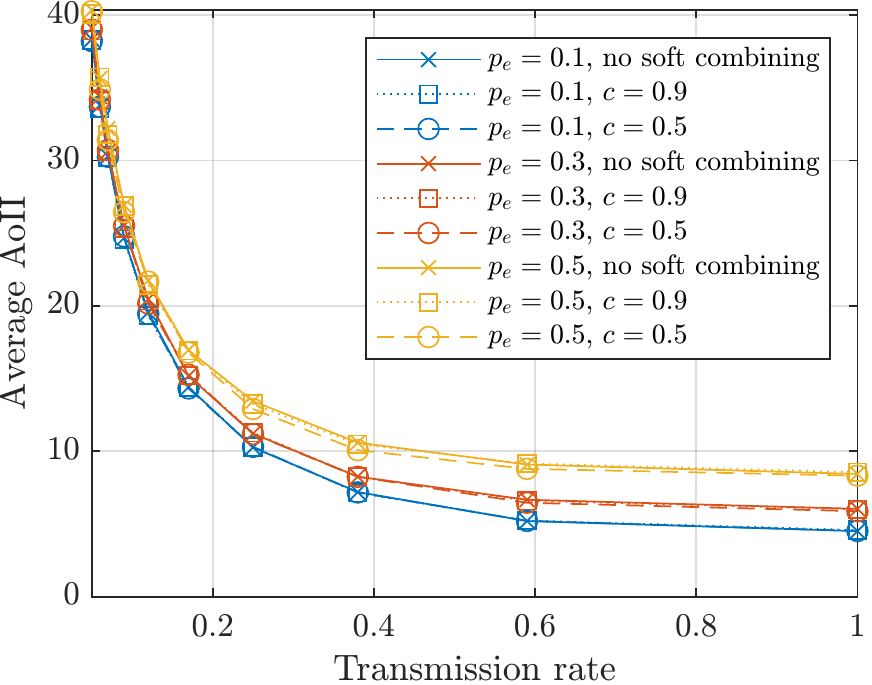}
    \caption{Average AoII versus the transmission rate constraint of the optimal threshold-based policy. The source model parameters are $\alpha=0.2$ and $N=128$ and the HARQ maximum re-transmissions are $r_{max}=2$.}
    \label{fig:p0_c_a_1}
    ~\\~\\
    \includegraphics[width=.94\linewidth]{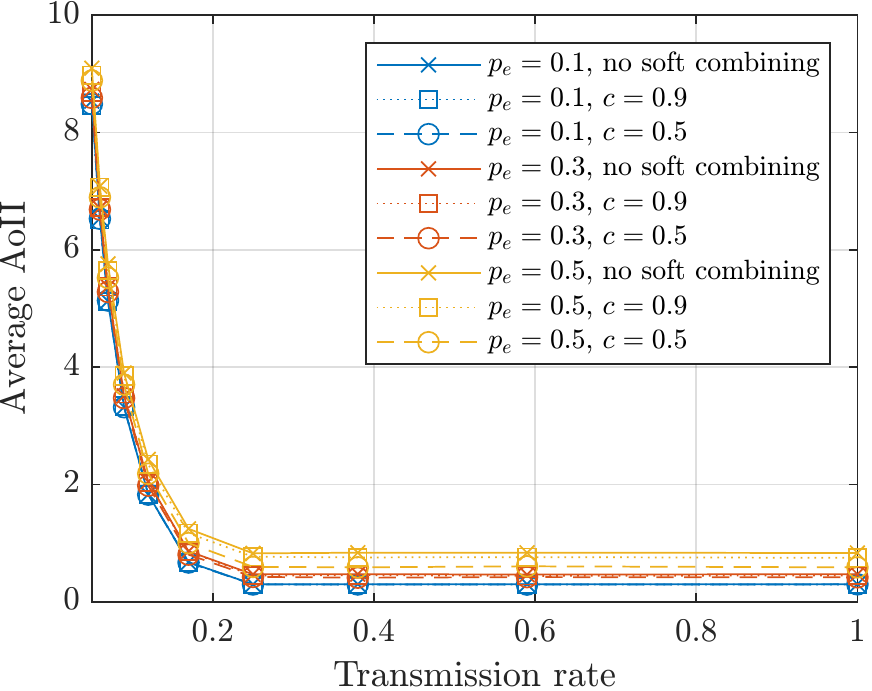}
    \caption{Average AoII versus the transmission rate constraint of the optimal threshold-based policy. The source model parameters are $\alpha=0.8$ and $N=128$ and the HARQ maximum re-transmissions are $r_{max}=2$.}
    \label{fig:p0_c_a_2}
\end{figure}

Figs.~\ref{fig:p0_r_N_1}-\ref{fig:p0_r_N_2} illustrate the average AoII for the same source models as before and for various $p_e$ and $r_{max}$ values when the decaying error rate constant $c$ is fixed to $0.5$. Again, soft combining is most helpful when the channel conditions are worse. Moreover, the gap between the cases with $r_{max}\!=\!2$ and $r_{max}\!=\!\infty$ is relatively small, corroborating the choice of small values in practical schemes~\cite{bib:harq_survey}.

Figs.~\ref{fig:p0_c_a_1}-\ref{fig:p0_c_a_2} illustrate the average AoII for two sources with the same $N$ but different $\alpha$ parameters. The results correspond to various $p_e$ and $c$ values, while the maximum re-transmissions are fixed to $r_{max}\!=\!2$. We observe that the achieved AoII increases when $\alpha$ gets smaller. Furthermore, in all experimental cases, we observe a point in the transmission rate axis beyond which the gains from increasing the transmission rate diminish. It can be seen that this point is decreasing as $\alpha$ increases.

\begin{figure}
    \centering
    \includegraphics[width=.94\linewidth]{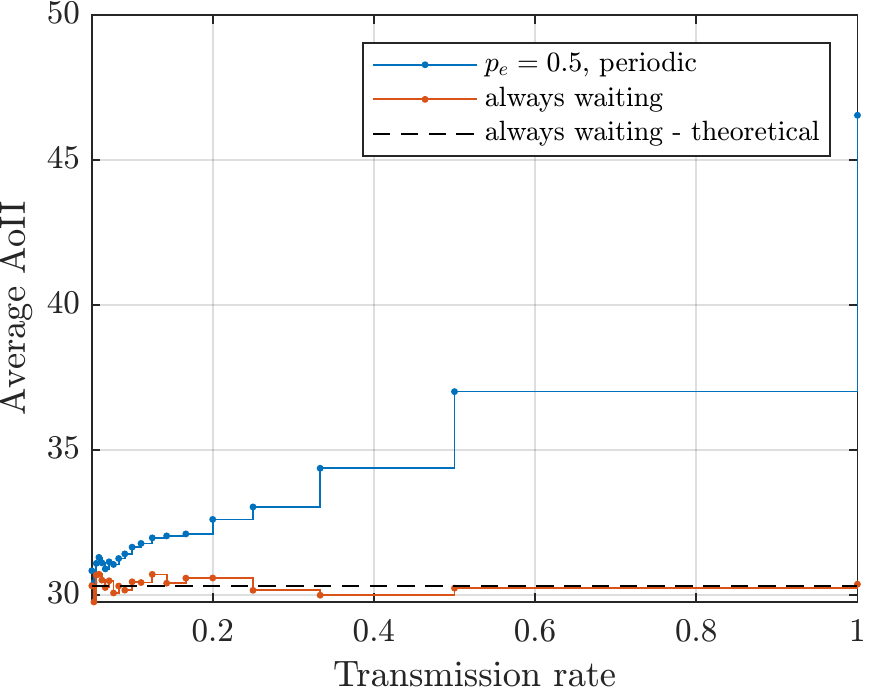}
    \caption{Average AoII versus the transmission rate constraint of the optimal waiting policy. The source model parameters are $\alpha=0.01$ and $N=32$.}
    \label{fig:always_wait_1}
    ~\\~\\
    \includegraphics[width=.94\linewidth]{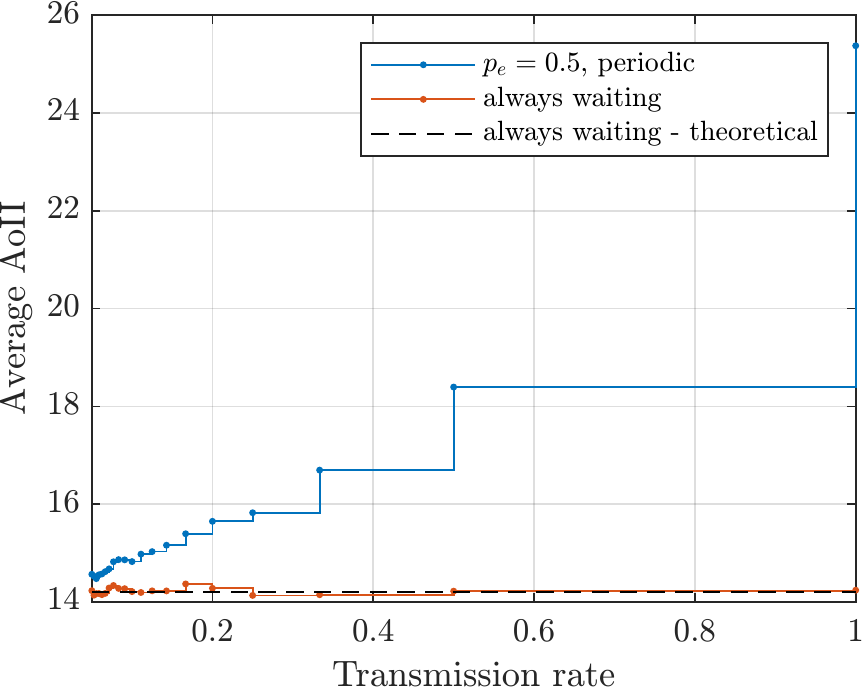}
    \caption{Average AoII versus the transmission rate constraint of the optimal waiting policy. The source model parameters are $\alpha=0.01$ and $N=16$.}
    \label{fig:always_wait_2}
\end{figure}

Our final results concern a case where the optimal policy is to always wait, i.e., when $\mu\!\geq\!\alpha$. Figs.~\ref{fig:always_wait_1}-\ref{fig:always_wait_2} illustrate the average AoII for two different sources with $\mu\!\geq\!\alpha$. The figures include both the simulated and the theoretical average AoII~\eqref{eq:waiting_AoII} of the optimal waiting policy. The optimal policy is compared to a periodic policy. It can be seen that, indeed, transmissions worsen the achieved AoII.

\section{Conclusion}\label{sec:conclusion}
This paper elaborated on the design of a remote monitoring system with HARQ, under a transmission rate constraint. The communication system was evaluated by the duration that the receiver has incorrect information for a remote $N$-ary symmetric Markov data source. To that end, we employed the long-term average AoII as the performance metric. The problem was formulated as a CMDP, and by leveraging its structural properties, we proved that an optimal transmission policy exists, which is a randomized mixture of two discrete threshold-based policies that randomize at most on one state. The optimal parameters are derived in analytic form and are computed using the binary search algorithm described in~\cite{bib:aoii}. Finally, the numerical results highlight the impact of (i) the source dynamics, (ii) the channel conditions and (iii) the resource constraint. 

Future extensions of this work that are of high interest, both from a theoretical and practical perspective, include the case of communication with delays that depend on the source state due to their different content, and more general Markov sources, e.g. where each state has a different probability to stay at the same state.

\appendix

\subsection{Derivation of the CMDP Transition Probabilities}\label{apdx:transition_prob}

First, consider the case where the transmitter opted to transmit, i.e., $y_t=1$. Let $D$ denote the event that the transmission results in successful decoding and with $\Bar{D}$ the complement of $D$. The events $D$ and $\Bar{D}$ happen with probability $p(r_t)$ and $1-p(r_t)$, respectively. When the event $D$ occurs, the transmission count $r_{t+1}$ always reverts to zero. When $\Bar{D}$ occurs, the transmission count $r_{t+1}$ increases only if the transmitted and the current source values coincide since the previous value will be re-transmitted. On the other hand, if the transmitted and current source values differ, the monitor will not benefit from a re-transmission of the old information since it is incorrect. Therefore, it is preferable to set the transmission count $r_{t+1}$ to zero and begin a new round of transmissions with fresh information.
%the next action is a re-transmission of the previous packet. If $\Bar{D}$ occurs and the next action is either waiting or a transmission with a new source value, the transmission count $r_{t+1}$ becomes zero.}

Assume that $\delta_t=0$ and the event $D$ occurs. We distinguish between two cases: i) the source has not changed between the time slots $t$ and $t{+}1$, and ii) the source has changed. In the first case, the decoded packet contains accurate information, while in the latter case, it is inaccurate. Note that in the second case, the AoII increases. Based on the source model, the first case occurs with probability $\alpha$ and the second with probability $1-\alpha$. Due to the independence of the source and the channel conditions, we get
\begin{align}
    &P\left(S_{t+1}{=}(0,0),D\mid\delta_t=0,r_t\geq 0, y_t = 1\right) = \alpha p(r_t),\label{eq:P00D}\\
    &P\left(S_{t+1}{=}(1,0),D\mid\delta_t=0,r_t\geq 0, y_t = 1\right) = (1-\alpha)p(r_t).\label{eq:P_10D}
\end{align}

Next, assume that $\delta_t=0$ and the event $\Bar{D}$ occurs. If the source has not changed between the time slots $t$ and $t{+}1$, which happens with probability $\alpha$, the AoII will remain zero. Also, the received packet's information and source value coincide, increasing the transmission count. These conditions yield

\begin{equation}\label{eq:P01BarD}
    P{\left(S_{t+1}{=}(0{,}r_t{+}1){,}\Bar{D}\mid\delta_t=0{,}r_t\geq 0, y_t = 1\right)} = \alpha (1-p(r_t)).
\end{equation}
On the other hand, if the source has changed, which happens with probability $1-\alpha$, the AoII increases to one. Also, the received packet contains inaccurate information compared to the source, which sets the transmission count to zero. Thus,

\begin{equation}\label{eq:P10BarD}
    P{\left(S_{t+1}{=}(1,0),\Bar{D}\mid\delta_t{=}0,r_t\geq 0, y_t {=} 1\right)} = (1-\alpha)(1-p(r_t)).
\end{equation}

Interestingly, by combining~\eqref{eq:P00D}-\eqref{eq:P10BarD} we obtain

\begin{align}
    &P\left(\delta_{t+1}{=}1\mid\delta_t=0,r_t\geq 0, y_t = 1\right) = 1-\alpha\,.\\
    &P\left(\delta_{t+1}{=}0\mid\delta_t=0,r_t\geq 0, y_t = 1\right) = \alpha\,.
\end{align}
In other words, $\delta_{t+1}$ is independent of $r_t$ under the condition that $\delta_t = 0$. In addition, we observe that if $\delta_{t+1}=1$, then, it always holds that $r_{t+1}=0$. On the other hand, if $\delta_{t+1}=0$, the future $\delta_{t+2}$ continues to be independent of $r_{t+1}$. Therefore, when $\delta_t=0$, the transmission count $r_t$ does not impact the future AoII. Hence, we shall ignore the transmission count when $\delta_t=0$ and freeze it to $r_t=0$. Thus,~\eqref{eq:P01BarD} can be rewritten as follows,

\begin{equation}\label{eq:P01BarD2}
    P\left(S_{t+1}{=}(0,0),\Bar{D}\mid\delta_t=0,r_t= 0, y_t = 1\right) = \alpha  (1-p(r_t)).
\end{equation}
By combining~\eqref{eq:P00D} with~\eqref{eq:P01BarD2} and \eqref{eq:P_10D} with~\eqref{eq:P10BarD}, we derive
\begin{align}
    &P\left(S_{t+1}{=}(0,0)\mid\delta_t=0,r_t\geq 0, y_t = 1\right) = \alpha\,.\\
    &P\left(S_{t+1}{=}(1,0)\mid\delta_t=0,r_t\geq 0, y_t = 1\right) = 1-\alpha\,.
\end{align}

Next, consider the case where $\delta_t>0$ and the transmitter opted to transmit, i.e., $y_t=1$. Assume the occurrence of event $D$, as defined above. If the source remained at the transmitted value after the reception of the packet, which happens with probability $\alpha$, the AoII becomes zero. However, if the source has changed state, the decoded information is incorrect, and the AoII increases. Due to the independence of the source and the channel conditions, we obtain
\begin{flalign}
    &P(S_{t+1}{=}(0,0),D\mid\delta_t>0,r_t\geq 0, y_t = 1) = \alpha p(r_t),\label{eq:P00BarD}\\ 
    &P{(S_{t+1}{=}(\delta_t{+}1{,}0),D\mid\delta_t\!>\!0,r_t\!\geq\! 0, y_t \!=\! 1)} = (1-\alpha) p(r_t).\label{eq:P10D}
\end{flalign}

Now assume that the event $\Bar{D}$ occurs. If the source has remained at the same state, which happens with probability $\alpha$, the AoII grows and the transmission count increases. If the source has changed and returned to the value already known by the monitor, which happens with probability $\mu$, the distortion becomes zero. Therefore, the AoII reverts to zero, and the same happens to the transmission count due to our observations on its independence from the future AoII when it is currently equal to zero. Lastly, in the case where the source changed but did not return to the previously known value, which happens with probability $1-\alpha-\mu$, the AoII increases and the transmission count reverts to zero, as the next transmission will contain the new value. Hence,
\begin{flalign}
    &P(S_{t+1}{=}(\delta_t{+}1{,}r_t{+}1),\Bar{D}\mid\delta_t{>}0,r_t{\geq} 0, y_t {=} 1)  = \alpha (1{-}p(r_t)),\label{eq:P11BarD_positiveAoII}\\
    &P(S_{t+1}{=}(0,0),\Bar{D}\mid\delta_t>0,r_t\geq 0, y_t = 1) =  \mu (1-p(r_t)),\label{eq:P00BarD2}\\
    &P(S_{t+1}{=}(\delta_t{+}1,0),\Bar{D}\mid\delta_t>0,r_t\geq 0, y_t = 1) \nonumber\\
    &= (1-\alpha-\mu)(1-p(r_t))\ .\label{eq:P10BarD_positiveAoII}
\end{flalign}
Combining~\eqref{eq:P00BarD} with~\eqref{eq:P00BarD2} and~\eqref{eq:P10D} with~\eqref{eq:P10BarD_positiveAoII} and utilizing the definitions in~\eqref{eq:gamma1}, \eqref{eq:gamma2}, we deduce that
\begin{flalign}
    &P(S_{t+1}{=}(0,0)\mid\delta_t>0,r_t\geq 0, y_t = 1) = 1\!-\!\gamma_1(r_t)\!-\!\gamma_2(r_t),\\
    &P(S_{t+1}{=}(\delta_t{+}1,0)\mid\delta_t>0,r_t\geq 0, y_t = 1) = \gamma_2(r_t).
\end{flalign}
Furthermore, we rewrite~\eqref{eq:P11BarD_positiveAoII} as
\begin{equation}
    P(S_{t+1}{=}(\delta_t{+}1,r_t{+}1)\mid\delta_t>0,r_t\geq 0, y_t = 1) = \gamma_1(r_t).
\end{equation}

Finally, we examine the case where the transmitter opted to wait, i.e., $y_t=0$. To that end, we will employ the symmetry of the source process. First, notice that the probability of arrival at some state equals $\mu$, regardless of the state of departure. Moreover, the probability of staying at a state equals $\alpha$ for all states. It follows that the distortion process in the absence of transmissions progresses according to a binary Markov chain as illustrated in Fig.~\ref{fig:distortion}. Then,

\begin{align}
    &P(S_{t+1}{=}(\delta_t{+}1,0)\mid\delta_t>0,r_t\geq 0, y_t = 0) = 1-\mu\,,\\
    &P(S_{t+1}{=}(0,0)\mid\delta_t>0,r_t\geq 0, y_t = 0) = \mu\,,\\
    &P(S_{t+1}{=}(1,0)\mid\delta_t=0,r_t= 0, y_t = 0) = 1-\alpha\,,\\
    &P(S_{t+1}{=}(0,0)\mid\delta_t=0,r_t= 0, y_t = 0) = \alpha\,.
\end{align}
Notice that the transmission count $r_{t+1}$ is always zero due to Remark~\ref{rem:wait_count}.

\begin{figure}
    \centering
    \includegraphics[width=\linewidth]{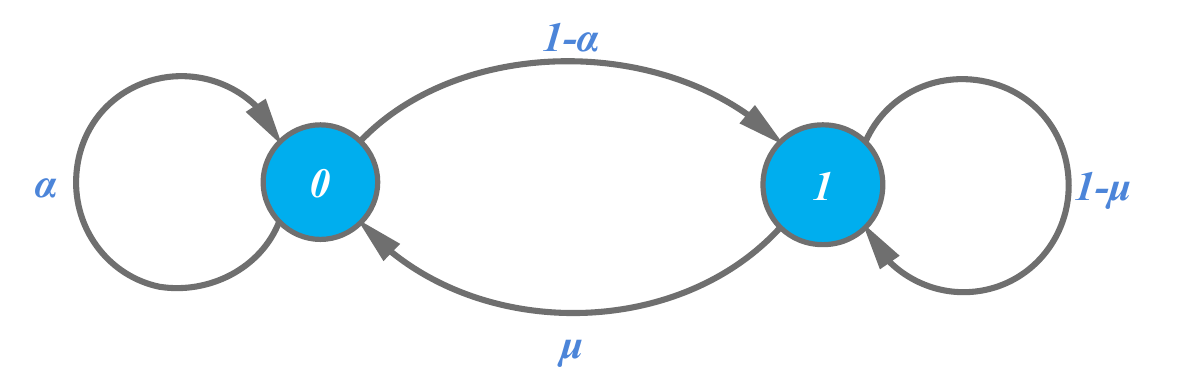}
    \caption{The distortion process in the absence of transmissions.}
    \label{fig:distortion}
\end{figure}

\subsection{Proof of Proposition~\ref{prop:rate}}\label{apdx:rate}

The state process of the threshold-based policy is the Markov Chain illustrated in Fig.~\ref{fig:threshold_MC}. Let us denote with $q_{\delta,r}$ the stationary distribution of the state $S=(\delta,r)$. The achieved transmission rate for a given threshold $n_0$ is the summation

\begin{equation}\label{eq:achiev_rate_1}
    C_{\pi_\lambda} = \sum_{0\leq h}\ \sum_{0\leq r\leq h} q_{n_0+h,r}\,.
\end{equation}

\begin{figure*}
    \centering
    \includegraphics[width=\linewidth]{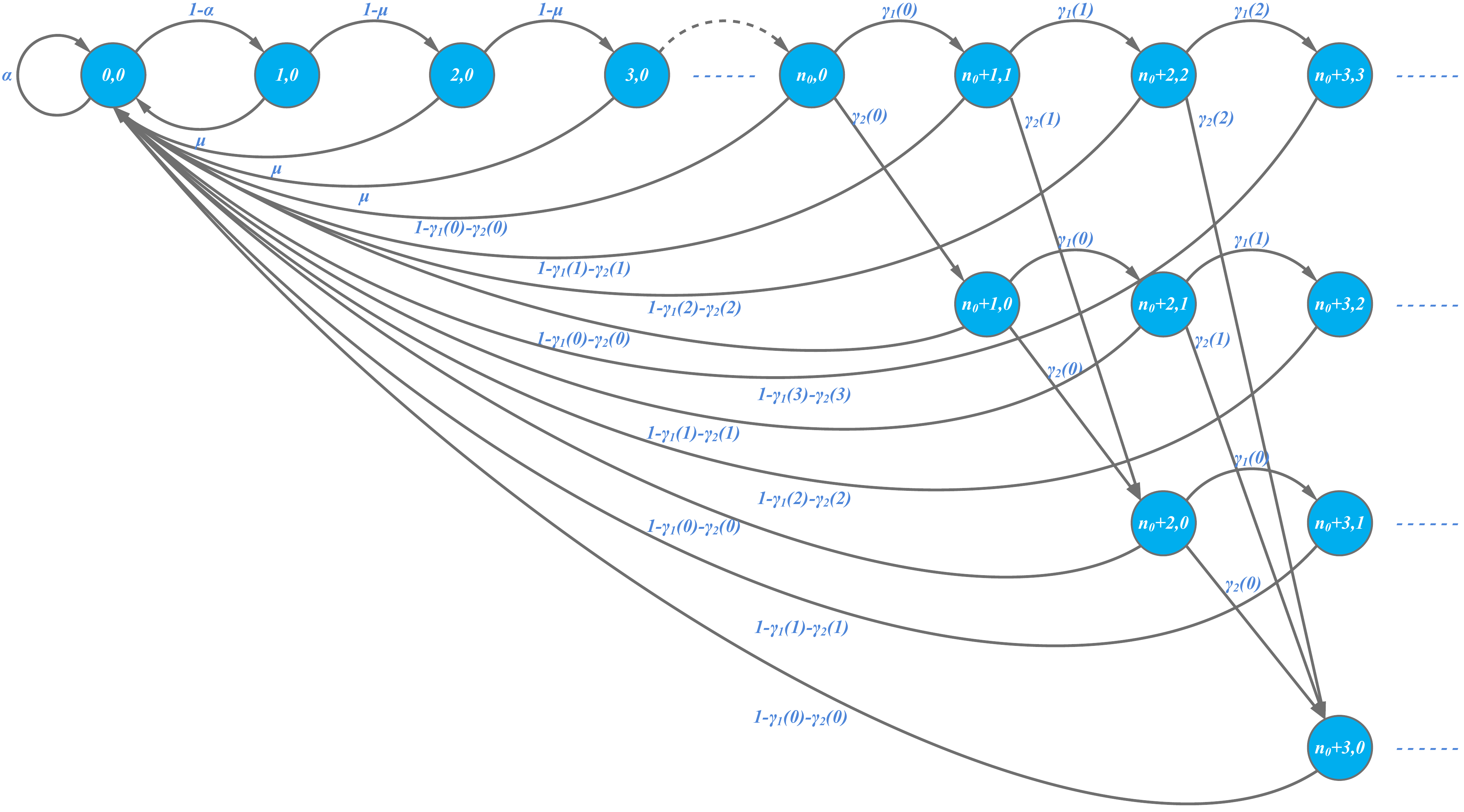}
    \caption{Transition process of the threshold-based transmission policy.}
    \label{fig:threshold_MC}
\end{figure*}

From Fig.~\ref{fig:threshold_MC} we infer the following equations,

\begin{align}
    &q_{n_0+h,r} = \begin{cases} q_{n_0+h-r,0} \prod_{j=0}^{r-1}\gamma_1(j) &\text{if $\ 0\leq h,\ 0\leq r\leq h$}\,,\\
                             0 &\text{if $\ 0\leq h,\ h < r$}\,,
                \end{cases} \label{eq:auxil1_prop3}\\
    &q_{n_0+h,0} = \sum_{k=0}^{h-1}\gamma_2(k)q_{n_0+h-1,k}\,, \quad \text{for $\ 1\leq h$}\,,\label{eq:auxil2_prop3}\ \\
    &q_{k,0} = (1-\alpha)(1-\mu)^{k-1}q_{0,0}\,, \quad \text{for $\ 1 \leq k \leq n_0$}\,. \label{eq:auxil3_prop3}
\end{align}
Combining~\eqref{eq:auxil1_prop3} and~\eqref{eq:auxil2_prop3} we derive

\begin{equation}\label{eq:auxil4_prop3}
    q_{n_0+h,0} = \sum_{k=0}^{h-1}\left[\gamma_2(k)q_{n_0+h-k-1,0} \prod_{j=0}^{k-1}\gamma_1(j)\right], \quad 1 \leq h\,.
\end{equation}

Define the function $m(h,r)$ that traces $q_{n_0+h,r}$ back to $q_{n_0,0}$ such that 

\begin{equation}
    q_{n_0+h,r} = m(h,r) q_{n_0,0}\,, \quad 0\leq h,\ 0\leq r\,.
\end{equation}
Equivalently, by~\eqref{eq:auxil3_prop3},
\begin{equation}\label{eq:prob_multipl_fac_all}
    q_{n_0+h,r} = (1-\alpha)(1-\mu)^{n_0-1} m(h,r) q_{0,0}\,, \quad 0\leq h,\ 0\leq r\,.
\end{equation}
Obviously, $m(0,0)=1$. Also, from~\eqref{eq:auxil1_prop3} we derive

\begin{equation}\label{eq:prob_multipl_fac_hr}
    m(h,r) = \begin{cases} m(h{-}r,0) \prod_{j=0}^{r-1}\gamma_1(j) &\text{if $\ 0\leq h,\ 0\leq r\leq h$}\,,\\
                             0 &\text{if $\ 0\leq h,\ h < r$}\,,
                \end{cases}
\end{equation}
and from~\eqref{eq:auxil4_prop3} 
\begin{equation}\label{eq:prob_multipl_fac_h0}
    m(h,0) = \sum_{k=0}^{h-1}\left[\gamma_2(k)m(h{-}k{-}1,0) \prod_{j=0}^{k-1}\gamma_1(j)\right], \quad 1 \leq h\,.
\end{equation}

Replacing~\eqref{eq:prob_multipl_fac_all} in~\eqref{eq:achiev_rate_1} we obtain~\eqref{eq:achiev_rate_result}. Lastly, the normalization condition of the stationary distribution states that
\begin{equation}\label{eq:q00_derivation}
    q_{0,0}+\sum_{k=1}^{n_0-1}q_{k,0} + \sum_{h=0}^\infty \sum_{r=0}^{h}q_{n_0+h,r} = 1\,,
\end{equation}
which after being combined with~\eqref{eq:auxil3_prop3} and~\eqref{eq:prob_multipl_fac_all}, yields the expression of $q_{0,0}$ in~\eqref{eq:q00}.

\subsection{Proof of Theorem~\ref{thm:cmdp_policy}}\label{apdx:cmdp_policy}

The proof is based on the results of~\cite{bib:sennott}. To this end, we will prove that assumptions 1-5 of~\cite{bib:sennott} hold for our problem. These assumptions are a set of structural properties of the CMDP that will allow proving the existence and structure of an optimal policy. First, we introduce a special class of policies with the following definition:

\begin{definition}[Definition 2.3 of~\cite{bib:sennott}]
Let $G\subset \mathcal{S}$ be a non-empty subset of states of a CMDP. Given a state $S\in\mathcal{S}$, let $\mathcal{R}(S,G)$ be the class of policies $\pi$ such that $P_{\pi}(S\in G \text{ for some } t\geq 1\mid S_0 = S)=1$ and the expected time $m_{S,G}$ of the first passage from $S$ to $G$ under $\pi$ is finite. Let $\mathcal{R^*}(S,G)$ be the class of policies $\pi\in \mathcal{R}(S,G)$ such that, in addition, the expected average AoII $J_{S,G}(\pi)$ and the expected transmission rate $C_{S,G}(\pi)$ of the first passage from $S$ to $G$ are finite.
\end{definition}

Next, we state and prove the assumptions mentioned above for our problem one by one.

\begin{assumption}
For all $b\!>\!0$, the set $G(b)\triangleq \{S\!=\!(\delta,r)\mid \text{ there exists an action } y \text{ such that } f(\delta)+y\leq b\}$ is finite.
\end{assumption}
\begin{proof}
This holds because the function $f(\cdot)$ is a monotonically increasing and unbounded function.
\end{proof}

\begin{assumption}There exists a deterministic policy $\pi$ that induces a Markov chain with the following properties: the state space $\mathcal{S}$ consists of a single (non-empty) positive recurrent class $K$ and a set $U$ of transient states such that $\pi\in\mathcal{R^*}(S,R)$, for any $S\in U$, and both the average AoII $J_{\pi}$ and the average transmission rate $C_\pi$ on $K$ are finite.
\end{assumption}
\begin{proof}
Consider the policy $\pi(S)\!=\!1$ for all $S\in\mathcal{S}$. That is, $\pi$ is the always-transmit policy. %From the transition probabilities in~\eqref{eq:transition_prob}, we can deduce that this policy induces a Markov chain that consists of a single recurrent class $R=\{(\delta,r):\delta=\mathbbm{N},r=\mathbbm{N}\}=\mathcal{S}$ (the transient set $U$ is empty).
By Lemma~\ref{lem:unichain}, the policy induces a Markov chain that consists of a single recurrent class $K=\{(\delta,r):\delta\in\mathbbm{N},r\in\mathbbm{N}\}=\mathcal{S}$ (the transient set $U$ is empty). Moreover, $C_\pi\!=\!1$ and also $J_{\pi}$ is finite due to the condition we imposed with~\eqref{eq:bound_cond}.
\end{proof}

\begin{assumption}Given any two states $S\neq S'\in \mathcal{S}$, there exists a policy $\pi$ (a function of $S$ and $S'$) such that $\pi \in \mathcal{R^*}(S,S')$.
\end{assumption}
\begin{proof}
Again, consider the policy $\pi(S)\!=\!1$ for all $S\in\mathcal{S}$. By Lemma~\ref{lem:unichain}, there is a positive probability to transit from $S$ to $S'$ and vice-versa. It is trivial to verify that for such a transition, the average AoII and the transmission rate of the first passage are finite, i.e. $\pi \in \mathcal{R^*}(S,S')$.
\end{proof}

\begin{assumption}
If a deterministic policy has at least one positive recurrent state, then it has a single positive recurrent class,
and this class contains the state $(0,0)$.
\end{assumption}
\begin{proof}
This stems in a straightforward way from Lemma~\ref{lem:unichain} and the fact that, from every state, there is a positive probability to transit to $(0,0)$.
\end{proof}

\begin{assumption}
There exists a policy $\pi$ such that $J_\pi \!<\! \infty$ and $C_\pi \!<\! R$.
\end{assumption}
\begin{proof}
Consider a threshold-based policy with threshold $n_0$, which induces a new transmission if and only if $\delta\!\geq\! n_0$. The threshold $n_0$ is chosen such that \mbox{$n_0 = \inf\{n_0\in\mathbbm{N} : C_{\pi} < R\}$}. Notice that $n_0$ is finite, since $C_{\pi}$ is decreasing with respect to $n_0$, as seen in~\eqref{eq:achiev_rate_result}. By the condition in~\eqref{eq:bound_cond}, we infer that $J_{\pi} \!<\! \infty$.
\end{proof}

Having proven assumptions 1-5, the results in~\cite[Thm. 2.5, Prop. 3.2, Lemma 3.9, Lemma 3.12]{bib:sennott} also hold for our CMDP, which proves Theorem~\ref{thm:cmdp_policy}.

\subsection{Proof of Proposition~\ref{prop:rate_monotonic}}\label{apdx:rate_monotonic}

Due to~\cite[Lemma 3.4]{bib:sennott}, the transmission rate $C_{\pi_{\lambda}}$ is non-increasing with respect to $\lambda$. Moreover, from the expression of $C_{\pi_\lambda}$ in~\eqref{eq:achiev_rate_result}, it is trivial to see that $C_{\pi_\lambda}$ is decreasing with respect to $n_0$. Since $C_{\pi_\lambda}$ is related with $\lambda$ only through $n_0$, it follows that $n_0$ is non-decreasing with respect to $\lambda$.

\bibliographystyle{IEEEtran}
\bibliography{IEEEabrv,biblio}

\begin{thebibliography}{10}
\providecommand{\url}[1]{#1}
\csname url@samestyle\endcsname
\providecommand{\newblock}{\relax}
\providecommand{\bibinfo}[2]{#2}
\providecommand{\BIBentrySTDinterwordspacing}{\spaceskip=0pt\relax}
\providecommand{\BIBentryALTinterwordstretchfactor}{4}
\providecommand{\BIBentryALTinterwordspacing}{\spaceskip=\fontdimen2\font plus
\BIBentryALTinterwordstretchfactor\fontdimen3\font minus \fontdimen4\font\relax}
\providecommand{\BIBforeignlanguage}[2]{{%
\expandafter\ifx\csname l@#1\endcsname\relax
\typeout{** WARNING: IEEEtran.bst: No hyphenation pattern has been}%
\typeout{** loaded for the language `#1'. Using the pattern for}%
\typeout{** the default language instead.}%
\else
\language=\csname l@#1\endcsname
\fi
#2}}
\providecommand{\BIBdecl}{\relax}
\BIBdecl

\bibitem{bib:aoi}
S.~Kaul, R.~Yates, and M.~Gruteser, ``Real-time status: How often should one update?'' in \emph{2012 Proceedings IEEE INFOCOM}, 2012, pp. 2731--2735.

\bibitem{bib:aoi_survey}
R.~D. Yates, Y.~Sun, D.~R. Brown, S.~K. Kaul, E.~Modiano, and S.~Ulukus, ``Age of information: An introduction and survey,'' \emph{IEEE Journal on Selected Areas in Communications}, vol.~39, no.~5, pp. 1183--1210, 2021.

\bibitem{bib:nonlinear_Kosta}
A.~Kosta, N.~Pappas, A.~Ephremides, and V.~Angelakis, ``The age of information in a discrete time queue: Stationary distribution and non-linear age mean analysis,'' \emph{IEEE Journal on Selected Areas in Communications}, vol.~39, no.~5, pp. 1352--1364, 2021.

\bibitem{bib:nonlinear_Sun}
Y.~Sun and B.~Cyr, ``Sampling for data freshness optimization: Non-linear age functions,'' \emph{Journal of Communications and Networks}, vol.~21, no.~3, pp. 204--219, 2019.

\bibitem{bib:value}
A.~Kosta, N.~Pappas, A.~Ephremides, and V.~Angelakis, ``Age and value of information: Non-linear age case,'' in \emph{2017 IEEE International Symposium on Information Theory (ISIT)}, 2017, pp. 326--330.

\bibitem{bib:updatewait}
Y.~Sun, E.~Uysal-Biyikoglu, R.~Yates, C.~E. Koksal, and N.~B. Shroff, ``Update or wait: How to keep your data fresh,'' in \emph{IEEE INFOCOM 2016 - The 35th Annual IEEE International Conference on Computer Communications}, 2016, pp. 1--9.

\bibitem{bib:stamatakis}
G.~Stamatakis, N.~Pappas, and A.~Traganitis, ``Control of status updates for energy harvesting devices that monitor processes with alarms,'' in \emph{2019 IEEE Globecom Workshops (GC Wkshps)}, 2019, pp. 1--6.

\bibitem{bib:wiener}
Y.~Sun, Y.~Polyanskiy, and E.~Uysal, ``Sampling of the {Wiener} process for remote estimation over a channel with random delay,'' \emph{IEEE Transactions on Information Theory}, vol.~66, no.~2, pp. 1118--1135, 2020.

\bibitem{bib:ornstein}
T.~Z. Ornee and Y.~Sun, ``Sampling and remote estimation for the {Ornstein-Uhlenbeck} process through queues: Age of information and beyond,'' \emph{IEEE/ACM Transactions on Networking}, vol.~29, no.~5, pp. 1962--1975, 2021.

\bibitem{bib:state1}
M.~Klügel, M.~H. Mamduhi, S.~Hirche, and W.~Kellerer, ``{AoI}-penalty minimization for networked control systems with packet loss,'' in \emph{IEEE INFOCOM 2019 - IEEE Conference on Computer Communications Workshops (INFOCOM WKSHPS)}, 2019, pp. 189--196.

\bibitem{bib:state2}
J.~P. Champati, M.~H. Mamduhi, K.~H. Johansson, and J.~Gross, ``Performance characterization using {AoI} in a single-loop networked control system,'' in \emph{IEEE INFOCOM 2019 - IEEE Conference on Computer Communications Workshops (INFOCOM WKSHPS)}, 2019, pp. 197--203.

\bibitem{bib:aoii}
A.~Maatouk, S.~Kriouile, M.~Assaad, and A.~Ephremides, ``The age of incorrect information: A new performance metric for status updates,'' \emph{IEEE/ACM Transactions on Networking}, vol.~28, no.~5, pp. 2215--2228, 2020.

\bibitem{bib:aoii_discrete}
Y.~Chen and A.~Ephremides, ``Minimizing age of incorrect information for unreliable channel with power constraint,'' in \emph{2021 IEEE Global Communications Conference (GLOBECOM)}, 2021, pp. 1--6.

\bibitem{bib:aoii_semantics2}
A.~Maatouk, M.~Assaad, and A.~Ephremides, ``The age of incorrect information: an enabler of semantics-empowered communication,'' \emph{IEEE Transactions on Wireless Communications}, pp. 1--1, 2022.

\bibitem{bib:aoii_pl}
S.~Saha, H.~Singh~Makkar, V.~Bala~Sukumaran, and C.~R. Murthy, ``On the relationship between mean absolute error and age of incorrect information in the estimation of a piecewise linear signal over noisy channels,'' \emph{IEEE Communications Letters}, vol.~26, no.~11, pp. 2576--2580, 2022.

\bibitem{bib:aoii_multiple_sources2}
S.~Kriouile and M.~Assaad, ``Minimizing the age of incorrect information for real-time tracking of markov remote sources,'' in \emph{2021 IEEE International Symposium on Information Theory (ISIT)}, 2021, pp. 2978--2983.

\bibitem{bib:aoii_multiple_sources1}
Y.~Chen and A.~Ephremides, ``Scheduling to minimize age of incorrect information with imperfect channel state information,'' \emph{Entropy}, vol.~23, no.~12, 2021.

\bibitem{9686027}
------, ``Minimizing age of incorrect information for unreliable channel with power constraint,'' in \emph{2021 IEEE Global Communications Conference (GLOBECOM)}, 2021, pp. 1--6.

\bibitem{bib:harq_survey}
A.~Ahmed, A.~Al-Dweik, Y.~Iraqi, H.~Mukhtar, M.~Naeem, and E.~Hossain, ``Hybrid automatic repeat request ({HARQ}) in wireless communications systems and standards: A contemporary survey,'' \emph{IEEE Communications Surveys Tutorials}, vol.~23, no.~4, pp. 2711--2752, 2021.

\bibitem{bib:harq_fading}
R.~Sassioui, M.~Jabi, L.~Szczecinski, L.~B. Le, M.~Benjillali, and B.~Pelletier, ``{HARQ} and {AMC}: Friends or foes?'' in \emph{2016 IEEE Global Communications Conference (GLOBECOM)}, 2016, pp. 1--7.

\bibitem{bib:aoi_harq}
E.~T. Ceran, D.~Gündüz, and A.~György, ``Average age of information with hybrid {ARQ} under a resource constraint,'' \emph{IEEE Transactions on Wireless Communications}, vol.~18, no.~3, pp. 1900--1913, 2019.

\bibitem{bib:threshold_policy_3}
A.~Arafa, J.~Yang, S.~Ulukus, and H.~V. Poor, ``Age-minimal transmission for energy harvesting sensors with finite batteries: Online policies,'' \emph{IEEE Transactions on Information Theory}, vol.~66, no.~1, pp. 534--556, 2020.

\bibitem{bib:threshold_policy_5}
B.~T. Bacinoglu, Y.~Sun, E.~Uysal, and V.~Mutlu, ``Optimal status updating with a finite-battery energy harvesting source,'' \emph{Journal of Communications and Networks}, vol.~21, no.~3, pp. 280--294, 2019.

\bibitem{bib:threshold_and_uniform_policy}
X.~Wu, J.~Yang, and J.~Wu, ``Optimal status update for age of information minimization with an energy harvesting source,'' \emph{IEEE Transactions on Green Communications and Networking}, vol.~2, no.~1, pp. 193--204, 2018.

\bibitem{bib:uniform_policy_1}
A.~Arafa and S.~Ulukus, ``Timely updates in energy harvesting two-hop networks: Offline and online policies,'' \emph{IEEE Transactions on Wireless Communications}, vol.~18, no.~8, pp. 4017--4030, 2019.

\bibitem{bib:uniform_policy_3}
S.~Feng and J.~Yang, ``Age of information minimization for an energy harvesting source with updating erasures: Without and with feedback,'' \emph{IEEE Transactions on Communications}, vol.~69, no.~8, pp. 5091--5105, 2021.

\bibitem{bib:bertsekas}
D.~P. Bertsekas, \emph{Dynamic Programming and Optimal Control, Vol. II}, 3rd~ed.\hskip 1em plus 0.5em minus 0.4em\relax Athena Scientific, 2007.

\bibitem{8778671}
B.~Zhou and W.~Saad, ``Joint status sampling and updating for minimizing age of information in the internet of things,'' \emph{IEEE Transactions on Communications}, vol.~67, no.~11, pp. 7468--7482, 2019.

\bibitem{8972306}
Q.~Wang, H.~Chen, Y.~Li, Z.~Pang, and B.~Vucetic, ``Minimizing age of information for real-time monitoring in resource-constrained industrial iot networks,'' in \emph{2019 IEEE 17th International Conference on Industrial Informatics (INDIN)}, vol.~1, 2019, pp. 1766--1771.

\bibitem{9085402}
M.~A. Abd-Elmagid, H.~S. Dhillon, and N.~Pappas, ``A reinforcement learning framework for optimizing age of information in rf-powered communication systems,'' \emph{IEEE Transactions on Communications}, vol.~68, no.~8, pp. 4747--4760, 2020.

\bibitem{bib:constant_delay_1}
I.~Krikidis, ``Average age of information in wireless powered sensor networks,'' \emph{IEEE Wireless Communications Letters}, vol.~8, no.~2, pp. 628--631, 2019.

\bibitem{bib:constant_delay_2}
S.~Leng and A.~Yener, ``Age of information minimization for an energy harvesting cognitive radio,'' \emph{IEEE Transactions on Cognitive Communications and Networking}, vol.~5, no.~2, pp. 427--439, 2019.

\bibitem{bib:constant_delay_3}
R.~D. Yates and S.~K. Kaul, ``Status updates over unreliable multiaccess channels,'' in \emph{2017 IEEE International Symposium on Information Theory (ISIT)}.\hskip 1em plus 0.5em minus 0.4em\relax IEEE Press, 2017, p. 331–335.

\bibitem{bib:constant_delay_4}
Y.-P. Hsu, E.~Modiano, and L.~Duan, ``Age of information: Design and analysis of optimal scheduling algorithms,'' in \emph{2017 IEEE International Symposium on Information Theory (ISIT)}, 2017, pp. 561--565.

\bibitem{bib:krishnamurthy}
V.~Krishnamurthy, \emph{Partially Observed Markov Decision Processes: From Filtering to Controlled Sensing}.\hskip 1em plus 0.5em minus 0.4em\relax Cambridge University Press, 2016.

\bibitem{bib:algorithms}
T.~H. Cormen, C.~E. Leiserson, R.~L. Rivest, and C.~Stein, \emph{Introduction to Algorithms, Third Edition}, 3rd~ed.\hskip 1em plus 0.5em minus 0.4em\relax The MIT Press, 2009.

\bibitem{bib:arq_per}
Q.~Liu, S.~Zhou, and G.~Giannakis, ``Cross-layer combining of adaptive modulation and coding with truncated {ARQ} over wireless links,'' \emph{IEEE Transactions on Wireless Communications}, vol.~3, no.~5, pp. 1746--1755, 2004.

\bibitem{bib:sennott}
L.~I. Sennott, ``Constrained average cost {M}arkov decision chains,'' \emph{Probability in the Engineering and Informational Sciences}, vol.~7, no.~1, pp. 69--83, 1993.

\end{thebibliography}

\vfill


\end{document}